\documentclass[pdflatex,sn-mathphys-ay]{sn-jnl}
\usepackage{amsmath,amssymb,amsfonts}%
\usepackage{amsthm}%
\usepackage{mathtools}%
\usepackage{rsfso}
\usepackage{bm}
\usepackage{graphicx}%
\graphicspath{{figures/}{./}}
\usepackage[shortlabels]{enumitem}%
\usepackage[title]{appendix}%
\usepackage{xurl}

\theoremstyle{thmstyleone}%
\newtheorem{theorem}{Theorem}[section]%
\newtheorem{proposition}{Proposition}[section]%

\theoremstyle{thmstylethree}%
\newtheorem{definition}{Definition}[section]%
\newtheorem{assumption}{Assumption}[section]%

\theoremstyle{thmstyletwo}%
\newtheorem{remark}{Remark}[section]%

\numberwithin{equation}{section}

\DeclareMathOperator{\diag}{diag}

\newcommand{\Strat}{\mathrm{Strat}}
\newcommand{\R}{\mathbb{R}}
\newcommand{\E}{\mathbb{E}}

\newcommand{\dd}{\mathrm{d}}

\makeatletter
\gdef\orcidlogo{\textnormal{\scriptsize\,[iD]}}
\makeatother

\makeatletter
\@ifpackageloaded{hyperref}{%
  \hypersetup{
    pdftitle={Consumption-investment under alternative stochastic-integral conventions},
    pdfauthor={Mario Ayala, Benjamin Vallejo Jimenez},
    pdfsubject={Merton consumption-investment problem under alpha-integral conventions},
    pdfkeywords={consumption-investment, stochastic integration conventions, Ito-Stratonovich, optimal portfolio choice, stochastic volatility, logarithmic utility}
  }}{}
\makeatother

\begin{document}

\title[Consumption--investment under alternative stochastic-integral conventions]{Consumption--investment under alternative stochastic-integral conventions}

\author[1]{\fnm{Mario} \sur{Ayala}\orcid{0000-0001-5717-275X}}\email{mario.ayala@tum.de}

\author*[2]{\fnm{Benjamin} \sur{Vallejo Jim\'enez}\orcid{0000-0002-7449-2939}}\email{benjamin@ucol.mx}

\affil[1]{\orgdiv{School of Computation, Information and Technology, Chair for Analysis and Modelling}, \orgname{Technische Universit\"at M\"unchen}, \orgaddress{\street{Boltzmannstra\ss e 3}, \city{Garching}, \postcode{85748}, \country{Germany}}}

\affil*[2]{\orgdiv{Facultad de Econom\'ia}, \orgname{Universidad de Colima}, \orgaddress{\street{Josefa Ortiz de Dom\'inguez \#64, Col.\ La Haciendita}, \city{Villa de \'Alvarez}, \postcode{28970}, \state{Colima}, \country{M\'exico}}}

\abstract{We study the Merton consumption--investment problem when risky asset
dynamics are written under stochastic--integral conventions indexed by
$\alpha$, interpolating between the It\^o ($\alpha=0$), Stratonovich
($\alpha=1/2$), and H\"anggi--Klimontovich ($\alpha=1$) conventions.
Transforming the drift jointly with the convention leaves the optimal policy
unchanged; holding the drift coefficient numerically fixed changes the
investor's opportunity set. Under logarithmic utility, consumption is
unaffected while the optimal portfolio weight shifts by $\alpha V^{-1}d_V$;
for a Heston asset, $\pi_t^\ast=\alpha+(\mu-r+\alpha\varrho\xi/2)/V_t$, which
can change sign across volatility states. Identifying the convention from
return data alone is impossible without additional structural restrictions.}

\keywords{Consumption--investment, Stochastic integration conventions, It\^o--Stratonovich interpolation, Optimal portfolio choice, Stochastic volatility, Logarithmic utility}

\pacs[JEL Classification]{G11, C61, D81}

\pacs[MSC Classification]{91G10, 93E20, 60H10, 91B16}

\maketitle

\section{Introduction}

Models of intertemporal choice under uncertainty often rely on continuous--time
representations of portfolio dynamics. Since the seminal contributions of
\citet{Merton1969,Merton1971}, the standard approach assumes that asset
returns follow Brownian diffusions and that investors choose consumption and
portfolio shares to maximize expected discounted utility. Under these premises,
the resulting optimal policies have shaped much of modern asset pricing,
household finance, and dynamic decision theory.

Beyond the classical formulation, numerous extensions of the Merton problem have
been proposed, including models with Markov--modulated drifts and volatilities
that capture regime--switching behaviour in financial markets
\citep{VallejoVenegasSoriano2015,VenegasMartinez2022,VallejoJimenez2017}. Further
developments allow for richer asset price dynamics, such as stochastic volatility
and jump--diffusion specifications, or more general time--varying investment
opportunities \citep{kraft2008optimal,Merton1973}. Other strands of the literature
introduce non--standard preferences and additional life--cycle features,
including habit formation, mortality, and insurance decisions
\citep{Merton1971,Constantinides1990,Richard1975}, or incorporate portfolio constraints
and transaction costs to better reflect institutional and regulatory frictions
\citep{DavisNorman1990,ShreveSoner1994}.

A key modelling convention shared by most of this literature is the use of the
It\^o integral to describe how information enters asset prices. The It\^o
integral evaluates integrands at the beginning of each time increment, which
matches the economic requirement that trading decisions be based on currently
available information; adaptedness itself, however, is a property of the class
of strategies and of the filtration, not of the integral notation. Indeed, an
equation written under a different convention can always be converted to an
equivalent It\^o semimartingale form, so the standard machinery of
arbitrage--free semimartingale finance continues to apply after conversion.

However, empirical research shows that high--frequency returns are affected by
microstructure frictions such as discreteness, bid--ask bounce, order--flow
imbalances, and latency
\citep{HansenLunde2006,AitSahaliaYu2009,CarteaJaimungalPenalva2015,FabozziFocardiJonas2011}.
In such environments, observed returns reflect information aggregated over short
intervals, and the stylized It\^o convention may not accurately represent how
recorded increments relate to the underlying continuous--time dynamics. This has
motivated interest in alternative stochastic interpretations, most notably the
Stratonovich integral and the more general $\alpha$--interpretation, which
evaluate integrands at different points within each time interval
\citep{DosReisPlatonov2021,YuanAo2012}. A change of convention by itself is a
change of representation: the same diffusion can be written under any
convention, provided the drift coefficient is transformed accordingly.
Substantive differences arise only when convention--specific coefficients are
held fixed while the convention varies, a distinction that is central to
this paper and that we formalize below.

Pre--averaging and other local smoothing procedures in the high--frequency
literature \citep{JacodEtAl2009} motivate asking whether particular
data--generation and sampling mechanisms may converge to effective stochastic
equations associated with different conventions. Establishing such a mapping
rigorously, and determining whether the resulting $\alpha$ is separately
identifiable from drift parameters, is beyond the scope of the present
paper; we return to this identification question in the concluding remarks.

Changing the interpretation of stochastic noise can also affect fundamental
structural properties of continuous--time models, such as time reversibility,
long--run invariant measures, and behaviour under temporal aggregation
\citep{ayala2025reversibility}. These features play an important role in how
stochastic dynamics are estimated, coarse--grained, or calibrated from
high--frequency data, and are therefore directly relevant for economic decision
problems built on estimated dynamics. At the same time, when asset prices are
driven by noises that are not semimartingales, classical self--financing
conditions may permit arbitrage. Fractional Brownian motion provides a prominent
example: \citet{cheridito2003arbitrage} shows that fractional Bachelier and
Samuelson models admit arbitrage unless admissible strategies are severely
restricted, while \citet{Bender2003} constructs a Stratonovich--type integral for
fractional Brownian motion and demonstrates arbitrage in fractional
Black--Scholes markets. Taken together, these results illustrate that seemingly
innocuous modelling conventions governing how randomness is interpreted can have
economically significant implications.

\medskip

\subsection*{Contribution}

If the drift coefficient is transformed jointly with the convention, the
underlying It\^o diffusion is unchanged, and with it the consumption--investment
problem and its optimal policy (Proposition~\ref{prop:invariance}). This is
immediate from the conversion formula, but it is worth stating, because it fixes
the meaning of everything that follows: no optimal policy can depend on the
stochastic--integral convention as such. This paper studies the complementary
exercise, in which convention--specific drift coefficients are held numerically
fixed. Changing the convention then alters the effective It\^o drift and hence
the investor's opportunity set, and the sensitivity results below quantify that
change: a specific change in the law of returns, not a representational artefact.

We first analyze markets with constant volatilities ( i.e. modeled by geometric Brownian motion) under the fixed--coefficient
experiment and show that the $\alpha$--interpretation then induces a transparent
deterministic shift in effective expected returns proportional to instantaneous
variances. For logarithmic utility, this yields closed--form optimal policies in
a market with $n$ risky assets. Optimal consumption remains a fixed proportion
of wealth, while optimal portfolio weights satisfy
\[
\theta_\alpha^*
  =
  V^{-1}\bigl(\mu - r\mathbf{1} + \alpha\,d_V\bigr),
\]
where $V=\Sigma\Sigma^\top$ is the return covariance matrix and
$d_V:=(V_{11},\dots,V_{nn})^\top$ collects its diagonal. The sensitivity
$\partial\theta_\alpha^*/\partial\alpha=V^{-1}d_V$ equals the vector of ones
exactly when the off--diagonal row sums of $V$ vanish, for $n\le3$ this
means $V$ diagonal, and from $n=4$ onwards a few further exceptions appear
(Remark~\ref{rem:sensitivity-direction}); away from that restriction the
per--asset correction generally departs from $\alpha$ and may even be
negative. In the single--asset case the rule
reduces to
\[
\theta_\alpha^* = \frac{\mu-r}{\sigma^2} + \alpha,
\]
so that, when the drift coefficient $\mu$ is held fixed, conventions further from
the It\^o benchmark prescribe higher optimal risky exposure.

We then consider asset volatility driven by a stochastic factor correlated with
returns. In this factor--driven setting the $\alpha$--dependent drift
adjustment decomposes into two channels: an own--variance component
$\alpha\,\sigma(X_t)^2$, which is present even when factor and return shocks
are uncorrelated, and a covariation component
$\alpha\,\varrho\,\sigma'(X_t)\nu(X_t)$ generated by the correlation between
factor and return noise. For logarithmic utility, candidate consumption remains
myopic, while the candidate risky fraction acquires both a constant and a
state--dependent $\alpha$--correction.

As a canonical illustration, we analyze a Heston stochastic volatility model,
introduced in~\citet{heston1993closed}, in which the instantaneous variance
follows a square--root diffusion and return and volatility shocks are
correlated. There the effective expected return becomes
$\mu_{\mathrm{eff}}(V_t)=\mu+\alpha\bigl(V_t+\varrho\xi/2\bigr)$, and the
long--run mean of the variance process shifts while its CIR structure is
preserved. For logarithmic utility, optimal consumption remains proportional to
wealth, while the optimal risky fraction takes the explicit form
\[
    c_t^\ast = \rho A_t,
    \qquad
    \pi_t^\ast
    =
    \frac{\mu_{\mathrm{eff}}(V_t)-r}{V_t}
    =
    \alpha + \frac{\mu-r+\alpha\varrho\xi/2}{V_t}.
\]
The correction thus combines a constant shift $\alpha$ with an
inverse--variance component whose sign depends on the correlation. Under the
empirically standard negative correlation $\varrho<0$ for equity indices, the
two components work in opposite directions, and the optimal exposure can change
sign across volatility states: in the calibration used below, the
H\"anggi--Klimontovich position is short at low variance levels and long above
an instantaneous volatility of $20\%$.

\medskip
Overall, our results show that, once nominal drift coefficients that are not
invariant across conventions are held fixed, the choice of stochastic--integral
convention has economically meaningful and state--dependent consequences for
optimal intertemporal decisions,  while a change of convention accompanied by
the corresponding drift transformation has none. This distinction also implies
an identification problem: with constant volatility, return data identify only
the combination $\mu+\alpha\sigma^2$, not $\mu$ and $\alpha$ separately, so
empirical statements about $\alpha$ require additional structural restrictions.

\subsection*{Structure of the paper}
Section~\ref{sec:math-prelim} introduces the probabilistic framework and the
stochastic--calculus conventions used throughout, and
Section~\ref{subsec:verification} states a verification theorem for logarithmic
utility that underlies all optimality claims made in the paper.
Section~\ref{subsec:merton-classical} revisits the classical one--asset Merton
consumption--investment problem under the It\^o interpretation, while
Section~\ref{sec:stratonovich} analyzes the same problem under the
Stratonovich convention and identifies the induced drift correction and its
impact on optimal portfolio choice.  
Section~\ref{sec:n-asset-alpha} extends the analysis to a market with $n$ risky
assets under a general $\alpha$--interpretation and derives closed--form optimal
consumption and portfolio rules together with their comparative statics in
$\alpha$.  
Section~\ref{sec:factor-correlated-noise} introduces factor--driven volatility
with correlated return and factor shocks, showing how the $\alpha$--interpretation
leads to genuinely state--dependent effects; the Heston stochastic volatility
model is treated as a canonical illustration.  
Appendix~\ref{app:dictionary} collects conversion formulas between It\^o,
Stratonovich, Klimontovich, and intermediate $\alpha$--conventions, including the
correlated--noise reduction used in the factor--driven setting.

\section{Preliminaries}

\subsection{Filtered probability space and Brownian motion}\label{sec:math-prelim}

We work on a complete probability space
\[
(\Omega,\mathcal{F},\mathbb{P})
\]
supporting an $n$-dimensional standard Brownian motion
\[
B=(B_t)_{t\ge 0}=(B_t^1,\dots,B_t^n)^{\top}.
\]
Thus $B_0=0$ almost surely, $B$ has continuous paths, and its increments are independent, stationary, and Gaussian with
\[
\mathbb{E}[B_t]=0,\qquad \mathrm{Cov}(B_t)=t\,I_n,\qquad t\ge 0,
\]
where $I_n$ is the $n\times n$ identity matrix.

Let $\{\mathcal{F}_t^B\}_{t\ge 0}$ be the natural filtration generated by $B$,
\[
\mathcal{F}_t^B := \sigma(B_s:\,0\le s\le t),
\]
and let $\{\mathbb{F}_t\}_{t\ge 0}$ denote its usual augmentation:
\[
\mathbb{F}_t := \bigcap_{u>t} (\mathcal{F}_u^B \vee \mathcal{N}),
\]
where $\mathcal{N}$ is the collection of $\mathbb{P}$-null sets contained in $\mathcal{F}$.
Throughout, $\{\mathbb{F}_t\}$ is assumed to be complete and right-continuous.  
All stochastic processes are $\mathbb{F}$-adapted unless explicitly stated otherwise.

All stochastic integrals below are the standard It\^o integral for
progressively measurable, locally square--integrable integrands with respect
to $\mathbb F$--Brownian motion, with the usual notation $\int_0^t H_s\,dB_s$,
quadratic (co-)variation $[\cdot]$, and It\^o isometry; see
\citet{KaratzasShreve1991} for the standard construction. For a symmetric
positive definite $n\times n$ matrix $R$, a Brownian motion $W$ with
covariance $[W]_t=R\,t$ can always be written $W_t=CB_t$ for a deterministic
$C$ with $CC^\top=R$ (e.g.\ a Cholesky factor); this fact underlies the
correlated--noise reductions used in Section~\ref{sec:factor-correlated-noise}
and Appendix~\ref{app:dictionary}. For a matrix $M$, $\mathrm{tr}(M)$ and
$\mathrm{diag}(M)$ denote its trace and the vector of its diagonal entries;
for vectors $x,y$, $x^\top y$ is the Euclidean inner product.

\begin{assumption}\label{ass:standing}
Throughout the paper, all processes are defined on a filtered probability
space $(\Omega,\mathcal F,(\mathcal F_t)_{t\ge0},\mathbb P)$ satisfying the
usual conditions, with all stochastic integrals taken with respect to
$(\mathcal F_t)$--Brownian motions. Drift and diffusion coefficients are
deterministic: whenever a state process $Y$ solves an SDE
\begin{equation}\label{eq:generic-alpha-sde}
    \dd Y_t = \mu(Y_t)\,\dd t + \sigma(Y_t)\,\circ_\alpha \dd W_t,
\end{equation}
we assume $\sigma\in C^1$ and that the coefficients of the \emph{equivalent
It\^o equation}, namely the converted drift
$\mu^{\mathrm{Ito}}:=\mu+\alpha\,\sigma\sigma'$ and the diffusion $\sigma$, 
are locally Lipschitz and of linear growth, and we require this for each
$\alpha$ considered in the comparison. These are the standard conditions
guaranteeing existence, uniqueness, and non-explosion of strong solutions;
see e.g.\ \citet{KaratzasShreve1991}. We state them on the converted equation
because that is the equation actually solved: every $\alpha$--interpretation
in this paper is given meaning through its It\^o form
(Proposition~\ref{prop:alpha-to-ito}), and imposing them instead on $\mu$,
$\sigma$, $\sigma'$ separately would not suffice. Linear growth of $\sigma$
and of $\sigma'$ does not imply linear growth of $\sigma\sigma'$: for
$\sigma(x)=x\sin x$ one has $\sigma'(x)=\sin x+x\cos x$, so
$|\sigma(x)|\le|x|$ and $|\sigma'(x)|\le1+|x|$ are both of linear growth,
whereas $\sigma(x)\sigma'(x)=x\sin^2x+x^2\sin x\cos x$ grows quadratically
along the sequence $x_k=\tfrac\pi4+2\pi k$ (where
$\sigma\sigma'=\tfrac12 x_k+\tfrac12 x_k^2$), so the cited conditions do not
apply to $\mu+\alpha\sigma\sigma'$. In the state--dependent multivariate case
the same requirement is imposed on the converted drift of
Proposition~\ref{prop:alpha-gamma-dictionary}. Section~\ref{sec:factor-correlated-noise}
relaxes local Lipschitz continuity at the boundary of the state space (as for
square--root diffusions), as formalized in Assumption~\ref{ass:factor}.
\end{assumption}

\subsection{A verification theorem for logarithmic utility}\label{subsec:verification}

Every optimality claim in this paper is obtained by the same route: guess a
value function of homothetic form, read off candidate optimal controls from the
first--order conditions of the Hamilton--Jacobi--Bellman equation, and check that
the resulting coefficients solve the HJB equation identically. This identifies a
\emph{candidate} optimal policy, but on its own does not establish optimality
against \emph{every} admissible control. We record the missing verification step
once, in a form general enough to cover Theorems~\ref{thm:Merton1D}
and~\ref{thm:n-asset-alpha} and the factor--driven controls of
Section~\ref{sec:factor-correlated-noise} as special cases, so that it need not be
repeated in each instance. This is a standard verification argument for
infinite--horizon discounted stochastic control problems; see e.g.\
\citet{FlemingSoner2006} for the general theory.

Consider wealth $a_t>0$ controlled jointly with an auxiliary scalar factor $x_t$
(possibly trivial, see Remark~\ref{rem:verification-trivial}
below),
\[
\dd a_t
  = a_t\Bigl(r+\theta_t^\top\bigl(\mu_{\mathrm{eff}}(x_t)-r\mathbf 1\bigr)
       -\frac{c_t}{a_t}\Bigr)\dd t
     + a_t\,\theta_t^\top\Sigma(x_t)\,\dd W_t,
\quad
\dd x_t = b(x_t)\,\dd t + \nu(x_t)\,\dd W_t^{x},
\]
where $\theta_t\in\R^n$, $W$ is the $n$--dimensional Brownian motion driving asset
returns, $\dd\langle W^i,W^x\rangle_t=\varrho_i\,\dd t$ for a correlation vector
$\varrho\in\R^n$, and $V(x):=\Sigma(x)\Sigma(x)^\top$ is positive definite for
every $x$.

\medskip
We now state our first verification theorem.

\begin{theorem}\label{thm:verification}
Suppose $v\in C^2$ solves
\begin{equation}\label{eq:verification-ode}
    \tfrac12\nu(x)^2v''(x)+b(x)v'(x)-\rho v(x)
    +\ln\rho-1+\frac{r}{\rho}
    +\frac{1}{2\rho}\bigl(\mu_{\mathrm{eff}}(x)-r\mathbf 1\bigr)^\top V(x)^{-1}\bigl(\mu_{\mathrm{eff}}(x)-r\mathbf 1\bigr)
    =0,
\end{equation}
and set $\phi(a,x):=\rho^{-1}\ln a+v(x)$ and $J(a,x,t):=\phi(a,x)\,e^{-\rho t}$. Let
\begin{equation}\label{eq:verification-controls}
    c_t^*=\rho\,a_t,
    \qquad
    \theta_t^*=V(x_t)^{-1}\bigl(\mu_{\mathrm{eff}}(x_t)-r\mathbf 1\bigr).
\end{equation}
Assume:
\begin{enumerate}[(i)]
\item for every admissible $(c_t,\theta_t)$ and every $T>0$,
\begin{equation}\label{eq:verification-martingale}
    \E\!\int_0^T
    \theta_t^\top V(x_t)\theta_t\,\dd t<\infty,
\end{equation}
so that the wealth stochastic integral appearing in the proof is a true
martingale on compact time intervals --- a restriction on the class of
controls, which we build into the admissible class
(Definition~\ref{def:admissible}(iii) below);
\item the function $v$ and the (autonomous) factor process satisfy, for
every $T>0$,
\begin{equation}\label{eq:verification-v-integrability}
    \E\!\int_0^T
    v'(x_t)^2\nu(x_t)^2\,\dd t<\infty,
\end{equation}
so that the factor stochastic integral in the proof is likewise a true
martingale on compact time intervals, and moreover the terminal factor value
has integrable negative part,
\begin{equation}\label{eq:verification-v-terminal}
    \E\bigl[v(x_T)^-\bigr]<\infty\qquad\text{for every }T>0;
\end{equation}
these are regularity requirements on $v$ along the factor dynamics, independent
of the controls. (Hypothesis~(i) alone does not deliver
\eqref{eq:verification-v-terminal}, since it controls the martingale integrand,
not the sign or size of $v$ at the terminal time.)
\item for every admissible $(c_t,\theta_t)$, with resulting wealth process
$a_t$, the terminal term satisfies
\begin{equation}\label{eq:verification-liminf}
    \liminf_{T\to\infty}\E\bigl[J(a_T,x_T,T)\bigr]\ge 0;
\end{equation}
\item $(c_t^*,\theta_t^*)$ is admissible, with resulting wealth process
$a_t^*$, and the transversality condition
\begin{equation}\label{eq:verification-transversality}
    \lim_{T\to\infty}\E\bigl[J(a_T^*,x_T,T)\bigr]=0
\end{equation}
holds.
\end{enumerate}
Then $(c_t^*,\theta_t^*)$ is optimal:
\[
    J(a_0,x_0,0)=\sup_{(c,\theta)\ \mathrm{admissible}}
    \E\left[\int_0^\infty e^{-\rho t}\ln c_t\,\dd t\right].
\]
\end{theorem}

\begin{remark}
Condition~\eqref{eq:verification-liminf} controls the terminal term of
\emph{every} admissible control, not only the candidate; without it, the
dynamic--programming inequality below bounds each control's value by
$J(a_0,x_0,0)$ minus a terminal term that could otherwise be strictly
negative in the limit.
\end{remark}

\begin{proof}
Since $\phi$ does not depend on $t$, It\^o's formula applied to
$J(a_t,x_t,t)=e^{-\rho t}\phi(a_t,x_t)$ along any admissible $(c_t,\theta_t)$ gives
\[
\dd J(a_t,x_t,t)
  = e^{-\rho t}\,\mathcal H(c_t,\theta_t;a_t,x_t)\,\dd t
    + e^{-\rho t}\,\phi_a\,a_t\theta_t^\top\Sigma(x_t)\,\dd W_t
    + e^{-\rho t}\,\phi_x\,\nu(x_t)\,\dd W_t^{x},
\]
where the martingale part is written as two separate integrals against the
$n$--dimensional return noise $W$ and the scalar factor noise $W^{x}$, and
\begin{multline*}
\mathcal H(c,\theta;a,x):=-\rho\phi(a,x)+\phi_a\,a\Bigl(r+\theta^\top(\mu_{\mathrm{eff}}(x)-r\mathbf1)-\frac{c}{a}\Bigr)\\
    +\tfrac12\phi_{aa}\,a^2\theta^\top V(x)\theta+\phi_x b(x)+\tfrac12\phi_{xx}\nu(x)^2;
\end{multline*}
the cross term $\phi_{ax}\,a\theta^\top\Sigma(x)\varrho\,\nu(x)$ is identically zero
because $\phi(a,x)=\rho^{-1}\ln a+v(x)$ is additively separable in $a$ and $x$, so
$\phi_{ax}\equiv0$ regardless of $\varrho$. Because $\ln c$ is strictly concave
in $c$ and the quadratic term
$\theta\mapsto\tfrac12\phi_{aa}a^2\theta^\top V(x)\theta$ appearing in
$\mathcal H$ is strictly concave in $\theta$ (indeed
$\phi_{aa}=-\rho^{-1}a^{-2}<0$ and $V(x)$ is positive definite, so this
quadratic form is negative definite), the first--order conditions in $c$
and $\theta$ are necessary and sufficient for a maximum of
$\ln c+\mathcal H(c,\theta;a,x)$, and a direct computation shows they are solved
exactly by~\eqref{eq:verification-controls}; substituting
\eqref{eq:verification-controls} back and using \eqref{eq:verification-ode} shows
the maximized expression vanishes identically. Hence, for every admissible
$(c,\theta)$,
\[
\ln c+\mathcal H(c,\theta;a,x) \le 0,
\]
with equality iff $(c,\theta)=(c^*(a),\theta^*(x))$ from
\eqref{eq:verification-controls}. Consequently
\[
\dd J(a_t,x_t,t) \le -e^{-\rho t}\ln c_t\,\dd t
    + e^{-\rho t}\,\phi_a\,a_t\theta_t^\top\Sigma(x_t)\,\dd W_t
    + e^{-\rho t}\,\phi_x\,\nu(x_t)\,\dd W_t^{x}.
\]
Note that $\phi_a(a,x)\,a=\rho^{-1}$, so under
conditions~\eqref{eq:verification-martingale}
and~\eqref{eq:verification-v-integrability} both stochastic integrals have
square--integrable integrands on $[0,T]$ and are therefore true martingales
there. Integrating from $0$ to $T$, taking expectations, and rearranging gives,
for every admissible $(c,\theta)$,
\[
\E\left[\int_0^T e^{-\rho t}\ln c_t\,\dd t\right]
  \le J(a_0,x_0,0)-\E\bigl[J(a_T,x_T,T)\bigr]
\]
(with the convention that the left--hand side may equal $-\infty$). This
expectation is well defined in $[-\infty,\infty)$: the pathwise form of the
preceding inequality, $\int_0^T e^{-\rho t}\ln c_t\,\dd t\le
J(a_0,x_0,0)-J(a_T,x_T,T)+N_T$ with $N_T$ the martingale part, bounds its
positive part by $|J(a_0,x_0,0)|+J(a_T,x_T,T)^-+|N_T|$, and
$\E|N_T|<\infty$ by~\eqref{eq:verification-martingale}--\eqref{eq:verification-v-integrability}.
Moreover $\E\bigl[J(a_T,x_T,T)^-\bigr]<\infty$ for $T$ large, because
$J(a_T,x_T,T)^-\le e^{-\rho T}\bigl(\rho^{-1}(\ln a_T)^-+v(x_T)^-\bigr)$ and the
two terms on the right are integrable by~\eqref{eq:adm-terminal} (which forces
$\E[(\ln a_T)^-]<\infty$ for $T$ large) and by the terminal
integrability~\eqref{eq:verification-v-terminal} of hypothesis~(ii); only the
negative part can diverge. Letting
$T\to\infty$ and invoking the liminf
condition~\eqref{eq:verification-liminf} yields
\[
\E\left[\int_0^\infty e^{-\rho t}\ln c_t\,\dd t\right]
  \le J(a_0,x_0,0)
\]
for every admissible control. Along the candidate control $(c^*,\theta^*)$ the
dynamic--programming inequality is an equality at every $t\ge0$; letting
$T\to\infty$ and invoking \eqref{eq:verification-transversality} yields
$\E[\int_0^\infty e^{-\rho t}\ln c_t^*\,\dd t]=J(a_0,x_0,0)$. Since the upper
bound holds for every admissible control, $(c^*,\theta^*)$ is optimal.
\end{proof}

\begin{remark}\label{rem:verification-trivial}
When the factor is absent (formally, $\nu\equiv0$, $b\equiv0$, and
$\mu_{\mathrm{eff}},\Sigma$ constant), equation~\eqref{eq:verification-ode}
degenerates to the algebraic identity
$-\rho v+\ln\rho-1+r/\rho+(2\rho)^{-1}(\mu_{\mathrm{eff}}-r\mathbf1)^\top V^{-1}(\mu_{\mathrm{eff}}-r\mathbf1)=0$,
solved by the constant
$v=\rho^{-1}(\ln\rho-1+r/\rho)+(2\rho^2)^{-1}(\mu_{\mathrm{eff}}-r\mathbf1)^\top V^{-1}(\mu_{\mathrm{eff}}-r\mathbf1)$,
and the transversality condition~\eqref{eq:verification-transversality} reduces to
$\lim_{T\to\infty}e^{-\rho T}\E[\ln a_T^*]=0$: indeed, with $v$ constant,
\[
J(a_T^*,x_T,T)=e^{-\rho T}\phi(a_T^*,x_T)
  =e^{-\rho T}\bigl(\rho^{-1}\ln a_T^*+v\bigr),
\]
so the discount factor enters exactly once, through the time argument of
$J$; it is not compounded with any discounting inside $\phi$, which is a
stationary function of $(a,x)$ alone. (Hypothesis~(ii) of
Theorem~\ref{thm:verification} holds trivially in this degenerate case: the
integral vanishes because $v'\equiv0$, and the terminal
requirement~\eqref{eq:verification-v-terminal} is immediate because $v$ is a
constant.) Since $a_T^*$ is a geometric
Brownian motion under the optimal policy
(Theorems~\ref{thm:Merton1D} and~\ref{thm:n-asset-alpha}), $\E[\ln a_T^*]$ grows
at most linearly in $T$, so this limit holds automatically. The same
observation also verifies the positive--part condition in
Definition~\ref{def:admissible}(v). Indeed,
$\ln c_t^*=\ln\rho+\ln a_t^*$, and the Gaussian law of $\ln a_t^*$ has mean
and variance growing at most linearly, whence
$\E[|\ln a_t^*|]\le C(1+t)$. Therefore
\[
\E\!\left[\int_0^\infty e^{-\rho t}(\ln c_t^*)^+\,\dd t\right]
\le \int_0^\infty e^{-\rho t}
       \bigl(|\ln\rho|+C(1+t)\bigr)\,\dd t<\infty.
\]
Likewise, the
liminf condition~\eqref{eq:verification-liminf} reduces to
$\liminf_{T\to\infty}e^{-\rho T}\E[\ln a_T]\ge0$ for every admissible control;
this is guaranteed by the no--doubt--worthless--wealth requirement
$\lim_{T\to\infty}e^{-\rho T}\E[(\ln a_T)^-]=0$ that we build into the class of
admissible controls below, which excludes policies whose discounted expected
log--wealth diverges to $-\infty$. This is precisely
the setting of Theorem~\ref{thm:Merton1D} ($n=1$) and
Theorem~\ref{thm:n-asset-alpha} (general $n$), and $v$ above coincides with
$\beta_0$ in each case. When the factor is present,
\eqref{eq:verification-ode} is a genuine second--order linear ODE for $v$; its
solution is discussed in Section~\ref{sec:factor-correlated-noise} for the case
relevant to that section.
\end{remark}

\subsection{Classical one–asset consumption–investment problem}\label{subsec:merton-classical}

In this subsection we recall the classical continuous–time consumption–investment
problem with logarithmic utility introduced by \citet{Merton1971}.
We work with $n=1$ assets, and on the filtered probability space
$(\Omega,\mathcal{F},\{\mathbb{F}_t\}_{t\ge 0},\mathbb{P})$
specified in Section~\ref{sec:math-prelim}.

In the classical formulation, an infinitely–lived agent observes the evolution of a 
financial market in continuous time and must choose, at each moment, (i) how much 
wealth to consume and (ii) how to split remaining wealth between a risk–free asset and 
a risky asset.  The objective is to maximize discounted expected utility of consumption 
over the infinite horizon.  A central modelling assumption is that decisions cannot 
anticipate future randomness: \textbf{controls must be adapted to the market filtration}.  
We also impose the usual self–financing constraint, so changes in wealth arise solely 
from investment gains and consumption.  Under these assumptions, Merton’s framework 
reduces the economic problem to a stochastic control problem driven by Brownian noise.

\subsubsection{Assets and source of randomness}
On the filtered probability space 
$(\Omega,\mathcal{F},\{\mathbb{F}_t\}_{t\ge 0},\mathbb{P})$
we consider a risk–free money–market account $b=(b_t)_{t\ge 0}$ and a risky
asset $S=(S_t)_{t\ge 0}$.  
The money–market account evolves deterministically as
\begin{equation}\label{eq:money-market}
\mathrm{d} b_t = r\,b_t\,\mathrm{d} t,
\qquad r>0,
\end{equation}
while the risky asset follows a geometric Brownian motion,
\begin{equation}\label{eq:1d-gbm-S}
\mathrm{d} S_t
    = \mu\, S_t\,\mathrm{d} t
    + \sigma\, S_t\,\mathrm{d} W_t,
\qquad \mu\in\mathbb{R},\ \sigma>0,
\end{equation}
where $W=(W_t)_{t\ge 0}$ is a one–dimensional standard Brownian motion 
adapted to $\{\mathbb{F}_t\}$.

\subsubsection{Controls and wealth dynamics}
We denote by $a_t$ the total wealth of the investor at time $t$.
A \emph{control} is a progressively measurable pair $(c_t,\theta_t)$, where
\begin{itemize}
    \item $c_t\ge 0$ is the consumption rate; the cumulative consumption process
    is $C_t=\int_0^t c_s\,\mathrm{d}s$,
    \item $\theta_t\in\mathbb{R}$ is the fraction of current wealth invested in
    the risky asset (so $1-\theta_t$ is invested in the money–market account).
\end{itemize}
The \emph{self–financing} condition stipulates that changes in wealth, whose value at time $t \geq0$ we denote by $a_t$,
are due only to trading gains/losses and consumption.  In differential form,
this reads
\begin{equation}\label{eq:wealth-identity}
\mathrm{d} a_t
  = a_t(1-\theta_t)\,\frac{\mathrm{d} b_t}{b_t}
  + a_t\theta_t\,\frac{\mathrm{d} S_t}{S_t}
  - c_t\,\mathrm{d} t.
\end{equation}
Substituting \eqref{eq:money-market} and \eqref{eq:1d-gbm-S} into
\eqref{eq:wealth-identity} yields the It\^o SDE for wealth,
\begin{equation}\label{eq:1d-wealth}
\mathrm{d} a_t
  = a_t\Big(r + \theta_t(\mu-r) - \frac{c_t}{a_t}\Big)\mathrm{d} t
  + a_t\,\theta_t\,\sigma\,\mathrm{d} W_t.
\end{equation}

\begin{definition}[Admissible controls]\label{def:admissible}
A progressively measurable pair $(c_t,\theta_t)$ is \emph{admissible} if:
\begin{enumerate}[(i)]
    \item $c_t> 0$ for Lebesgue--a.e.\ $t\ge0$ almost surely (so that
          $\ln c_t$ is well defined) and $\theta_t\in\mathbb{R}$, with
          $C_t=\int_0^t c_s\,\mathrm{d} s$ finite almost surely;
    \item the wealth process $a_t$ is strictly positive and satisfies the SDE
          \eqref{eq:1d-wealth};
    \item for every $T>0$,
          \begin{equation}\label{eq:int-cond}
          \mathbb{E}\!\int_0^T \theta_t^2\sigma^2\,\mathrm{d} t < \infty;
          \end{equation}
    \item the discounted expected log--wealth does not collapse:
          \begin{equation}\label{eq:adm-terminal}
          \lim_{T\to\infty} e^{-\rho T}\,\mathbb{E}\bigl[(\ln a_T)^-\bigr]=0;
          \end{equation}
    \item the discounted positive part of running utility is integrable:
          \begin{equation}\label{eq:adm-utility-pos}
          \mathbb{E}\!\int_0^\infty e^{-\rho t}(\ln c_t)^+\,\mathrm{d}t<\infty,
          \end{equation}
          so that the objective~\eqref{eq:1d-objective} is well defined in
          $[-\infty,\infty)$ and cannot take the indeterminate form
          $+\infty-\infty$.
\end{enumerate}
In the $n$--asset and factor--driven settings of
Sections~\ref{sec:n-asset-alpha} and~\ref{sec:factor-correlated-noise},
condition \eqref{eq:int-cond} is replaced by its analogue
$\mathbb{E}\int_0^T \theta_t^\top V(x_t)\theta_t\,\mathrm{d}t<\infty$, which is
hypothesis~(i) of Theorem~\ref{thm:verification} verbatim.
\end{definition}

\begin{remark}
These conditions jointly ensure economic coherence and mathematical
well--posedness: $c_t>0$ with $C_t$ finite keeps consumption economically
interpretable and the logarithmic utility well defined; $a_t>0$ solving
\eqref{eq:1d-wealth} guarantees existence, uniqueness, and non-explosion of
wealth, preventing bankruptcy in finite time; the integrability
condition~\eqref{eq:int-cond} rules out explosive cumulative risk exposure
(it does not impose a pointwise bound on leverage) and is precisely the
martingale hypothesis~\eqref{eq:verification-martingale} of
Theorem~\ref{thm:verification}, so that the verification argument of
Section~\ref{subsec:verification} applies to every admissible control. It
also keeps the wealth equation itself well posed: since $a_t$ is continuous,
it is pathwise bounded on compact intervals, so
$\int_0^T a_t^2\theta_t^2\sigma^2\,\mathrm{d}t
\le\bigl(\sup_{t\le T}a_t^2\bigr)\,\sigma^2\int_0^T\theta_t^2\,\mathrm{d}t
<\infty$ almost surely by~\eqref{eq:int-cond}, and the stochastic integral
in~\eqref{eq:1d-wealth} is a well--defined continuous local martingale.
Note that no integrability of $c_t$ beyond the pathwise finiteness of $C_t$
in~(i) is required: consumption enters the verification argument only through
the term $\phi_a\,c=c/(\rho a)$ of the Hamiltonian, and since $a_t$ is
continuous and strictly positive, $\inf_{t\le T}a_t>0$ almost surely, so
$\int_0^T c_t/a_t\,\mathrm{d}t\le(\inf_{t\le T}a_t)^{-1}C_T<\infty$ almost
surely. The stronger condition $\E\int_0^T c_t\,\mathrm{d}t<\infty$ would make
the discounted \emph{utility} positive part integrable directly; we impose
only the weaker~\eqref{eq:adm-utility-pos}, a bound on
$(\ln c_t)^+$ rather than on $c_t$, which is what the objective actually
requires and which need not be verified through a separate moment of
consumption. Condition~\eqref{eq:adm-utility-pos} guarantees that
$\E[\int_0^\infty e^{-\rho t}\ln c_t\,\mathrm{d}t]$ is well defined in
$[-\infty,\infty)$ for every admissible control, so that the value
function~\eqref{eq:1d-objective} and the optimality comparison of
Theorem~\ref{thm:verification} never confront an indeterminate
$+\infty-\infty$.
Finally, the terminal condition~\eqref{eq:adm-terminal} controls the terminal
term in the verification argument
(condition~\eqref{eq:verification-liminf} of Theorem~\ref{thm:verification}).
\end{remark}

\subsubsection{Utility maximization problem}
Given a subjective discount rate $\rho>0$, the investor selects an admissible 
control $(c_t,\theta_t)$ to maximize the expected discounted utility of 
consumption.  The associated \emph{value function} is defined by
\begin{equation}\label{eq:1d-objective}
J(a,t)
  := \sup_{(c,\theta)\ \text{adm.}}
     \mathbb{E}\!\left[
        \int_t^\infty e^{-\rho s}\,u(c_s)\,\mathrm{d}s
        \,\Big|\, a_t=a
     \right],
\end{equation}
where $a>0$ denotes current wealth at time $t$, and where throughout this
subsection we take $u(x)=\ln x$. Note that utility is discounted from calendar
time $0$, so that in this time--homogeneous environment
$J(a,t)=e^{-\rho t}\,\widetilde V(a)$ for a stationary function $\widetilde V$;
we keep the calendar--time formulation because it matches the verification
theorem of Section~\ref{subsec:verification}.
Under logarithmic utility the investor exhibits constant relative risk aversion,
and the optimal portfolio rule depends only on the instantaneous investment
opportunities (the so-called \emph{myopic} property).

To characterise the function $J$, we assume that the dynamic programming 
principle holds and that $J$ is sufficiently smooth: specifically, 
$J\in C^{2,1}((0,\infty)\times[0,\infty))$, with 
\[
J_a(a,t)=\frac{\partial J}{\partial a}(a,t), \qquad
J_{aa}(a,t)=\frac{\partial^2 J}{\partial a^2}(a,t), \qquad
J_t(a,t)=\frac{\partial J}{\partial t}(a,t).
\]
Applying It\^o's formula to $J(a_t,t)$ along the wealth dynamics 
\eqref{eq:1d-wealth} yields that $J$ must satisfy the Hamilton--Jacobi--Bellman 
equation
\begin{equation}\label{eq:1d-hjb}
0 = \sup_{c> 0,\ \theta\in\mathbb{R}}
\left\{
e^{-\rho t}\ln c + J_t(a,t)
+ J_a(a,t)\,a\!\left(r+\theta(\mu-r)-\frac{c}{a}\right)
+ \tfrac12 J_{aa}(a,t)a^2\theta^2\sigma^2
\right\}.
\end{equation}
Note that the discounting enters only through the factor $e^{-\rho t}$
multiplying the running utility and through the time derivative $J_t$; for
$J(a,t)=e^{-\rho t}\widetilde V(a)$ one has $J_t=-\rho J$, and
\eqref{eq:1d-hjb} reduces to the stationary equation
$0=\sup_{c,\theta}\{\ln c-\rho\widetilde V+\widetilde V'(a)a(r+\theta(\mu-r)-c/a)+\tfrac12\widetilde V''(a)a^2\theta^2\sigma^2\}$.

Because $\ln$ is strictly concave in $c$ and $J_{aa}(a,t)< 0$ for the candidate
value function below (the value function is strictly concave in wealth), the
maximizers of \eqref{eq:1d-hjb} are characterised by the corresponding
first–order conditions.

\smallskip
We now state the Merton solution for the case of one risky asset.
\begin{theorem}\label{thm:Merton1D}
Consider the wealth dynamics \eqref{eq:1d-wealth} under admissible controls
$(c_t,\theta_t)$, and let the utility function be $u(x)=\ln x$. The unique
optimal controls are
\[
c_t^* = \rho\,a_t,
\qquad
\theta_t^* = \frac{\mu - r}{\sigma^2},
\]
and the corresponding value function is
\begin{equation}\label{eq:1d-value}
J(a,t)
= \left(\beta_0 + \frac{1}{\rho}\ln a \right)e^{-\rho t},
\end{equation}
where
\begin{equation}
\beta_0
= \frac{1}{\rho}\Bigl(\frac{r}{\rho} + \ln \rho - 1\Bigr)
  + \frac{1}{2\rho^2}\,\frac{(\mu-r)^2}{\sigma^2}.
\end{equation}
\end{theorem}

\begin{proof}
This is the $n=1$, factor--absent case of Theorem~\ref{thm:verification}
(Remark~\ref{rem:verification-trivial}), with $\mu_{\mathrm{eff}}=\mu$ and
$V=\sigma^2$. The candidate value function
$J(a,t)=(\beta_0+\rho^{-1}\ln a)e^{-\rho t}$ solves the HJB
equation~\eqref{eq:1d-hjb} by direct substitution: with
$J_a=\rho^{-1}a^{-1}e^{-\rho t}$, $J_{aa}=-\rho^{-1}a^{-2}e^{-\rho t}$, and
$J_t=-\rho J$, the first--order condition in $c$ reads
$e^{-\rho t}/c=J_a$, i.e.\ $c^*=\rho a$, and the first--order condition in
$\theta$ gives $\theta^*=(\mu-r)/\sigma^2$; substituting these back forces
$\beta_0$ to take the stated value, making the right--hand side of
\eqref{eq:1d-hjb} vanish identically. As for the hypotheses of
Theorem~\ref{thm:verification}: hypothesis~(i)
(condition~\eqref{eq:verification-martingale}) holds for \emph{every}
admissible control directly by Definition~\ref{def:admissible}(iii), and
hypothesis~(ii) is vacuous here since the factor is absent and $v$ is
constant. Admissibility of $(c^*,\theta^*)$ --- note that $\theta^*$ is
constant, so \eqref{eq:int-cond} is immediate --- and the transversality
condition~\eqref{eq:verification-transversality} hold by
Remark~\ref{rem:verification-trivial}, since wealth follows a geometric Brownian
motion under this control (see the remark below). Optimality now follows from
Theorem~\ref{thm:verification}.
\end{proof}

\begin{remark}
In Merton's original formulation \citeyearpar{Merton1971}, the logarithmic case
appears as a special instance of the general HARA (hyperbolic absolute risk 
aversion) utility class.  
For this specific utility, the HJB equation \eqref{eq:1d-hjb} admits a closed-form 
solution, and the first–order conditions lead directly to the explicit controls 
and value function stated in Theorem~\ref{thm:Merton1D}.  
Several modern expositions provide concise derivations of these formulas.
\end{remark}

\begin{remark}
Substituting $c_t^*=\rho\,a_t$ and $\theta_t^*=(\mu-r)/\sigma^2$ into the wealth 
equation \eqref{eq:1d-wealth} yields the geometric SDE
\[
\frac{\mathrm{d}a_t}{a_t}
  = \left(r + \frac{(\mu-r)^2}{\sigma^2} - \rho\right)\mathrm{d}t
    + \frac{\mu-r}{\sigma}\,\mathrm{d}W_t.
\]
Applying It\^o's formula to $\ln a_t$ gives
\[
\mathrm{d}(\ln a_t)
  = \left(r - \rho + \frac12\,\frac{(\mu-r)^2}{\sigma^2}\right)\mathrm{d}t
    + \frac{\mu-r}{\sigma}\,\mathrm{d}W_t,
\]
which is consistent with the value function \eqref{eq:1d-value}.
Moreover, under the optimal policy the transversality
condition~\eqref{eq:verification-transversality} of
Theorem~\ref{thm:verification},
\[
\lim_{T\to\infty}\mathbb{E}\!\left[J(a_T^*,T)\right]
  =\lim_{T\to\infty}e^{-\rho T}\,\mathbb{E}\!\left[\beta_0+\rho^{-1}\ln a_T^*\right]
  =0,
\]
holds, ensuring that the infinite-horizon optimization problem is well posed:
by \eqref{eq:1d-value} the discount factor is already carried by $J$ through
its time argument, and $\E[\ln a_T^*]$ grows at most linearly in $T$ by the
displayed dynamics above.
\end{remark}

\begin{remark}
In the classical formulation above, the economic environment is fully characterised 
by three modelling choices: the self–financing wealth equation, the restriction to 
non–anticipative (i.e., $\mathbb{F}_t$–adapted) controls, and the adoption of the 
It\^o interpretation for the stochastic integral.  These assumptions jointly determine 
the wealth dynamics and, through the associated HJB equation, the resulting optimal 
consumption and portfolio policies.  In particular, the familiar Merton portfolio 
share arises precisely from this combination of conventions.

In the remainder of the paper we illustrate that, once the interpretation of the
noise is allowed to vary, moving from It\^o to Stratonovich, Klimontovich, or
the general $\alpha$--scheme, \emph{while the nominal drift coefficient
$\mu$ is held numerically fixed}, the optimal investment rule changes in a
systematic and quantifiable manner. The reason is that holding the
convention--specific coefficient fixed while changing the convention changes
the effective It\^o drift of the asset, and with it the law of returns and the
investor's opportunity set. This is a statement about parametrization, not
about information: the filtration, the adaptedness of controls, and the
self--financing rule are identical across conventions, and if the drift
coefficient is transformed jointly with the convention the optimal policy does
not change at all (Proposition~\ref{prop:invariance} below). For $\alpha>0$
and fixed nominal $\mu$, the induced drift correction increases the effective
return of the risky asset and leads, under logarithmic utility, to a larger
optimal portfolio share.
\end{remark}

\section{Stratonovich interpretation and its impact on optimal investment}
\label{sec:stratonovich}

We now reconsider the one–asset consumption–investment problem when the risky
asset is modelled with the Stratonovich interpretation of noise, corresponding
to $\alpha=\tfrac12$ in the general $\alpha$–scheme.

\subsection{Stratonovich GBM and effective drift}

Let the risky asset $S=(S_t)_{t\ge0}$ satisfy
\begin{equation}\label{eq:Strat-S}
\frac{\dd S_t}{S_t}
    = \mu\,\dd t + \sigma\,\circ_{1/2}\dd W_t,
    \qquad \mu\in\R,\ \sigma>0,
\end{equation}
while the money–market account evolves as in \eqref{eq:money-market}.
A Stratonovich SDE of the form
\[
\dd X_t=b(X_t)\,\dd t + \sigma(X_t)\,\circ_{1/2}\dd W_t
\]
is equivalent, using Proposition~\ref{prop:alpha-gamma-dictionary}, to the Itô SDE
\[
\dd X_t
    = \bigl(b(X_t)+\tfrac12\sigma(X_t)\sigma'(X_t)\bigr)\dd t
      +\sigma(X_t)\,\dd W_t.
\]
Since $\sigma(x)=\sigma x$ in \eqref{eq:Strat-S}, we obtain the equivalent Itô form
\begin{equation}\label{eq:Strat-S-Ito}
\frac{\dd S_t}{S_t}
    = \bigl(\mu+\tfrac12\sigma^2\bigr)\dd t + \sigma\,\dd W_t.
\end{equation}
Thus the Stratonovich interpretation does not merely change the calculus rules:
for the same parameters $(\mu,\sigma)$ it produces a higher \emph{effective drift}
\[
\mu_{\mathrm{eff}}^{(1/2)}=\mu+\tfrac12\sigma^2.
\]

\subsection{Wealth dynamics and the origin of the modified Merton rule}

Throughout the paper, trading gains are defined by the standard It\^o
(semimartingale) integral against the asset price: the Stratonovich writing of
the price process~\eqref{eq:Strat-S} is first converted to its equivalent
It\^o form~\eqref{eq:Strat-S-Ito}, and the self--financing identity
\eqref{eq:wealth-identity} is then applied to that semimartingale. This keeps
the definition of gains, and hence of wealth, identical across all
conventions considered in the paper; only the effective drift of the price
changes. With $a_t>0$ denoting wealth, $c_t\ge0$ consumption, and
$\theta_t\in\R$ the risky fraction, this yields the wealth SDE
\begin{equation}\label{eq:wealth-Strat}
\dd a_t
    = a_t\Bigl(r + \theta_t\bigl(\mu+\tfrac12\sigma^2-r\bigr)
               -\frac{c_t}{a_t}\Bigr)\dd t
      + a_t\theta_t\sigma\,\dd W_t,
\end{equation}
which differs from the classical Itô formulation only through the drift adjustment
$\mu-r\mapsto \mu+\tfrac12\sigma^2-r$.

Crucially, \emph{all steps of the dynamic programming argument remain unchanged}.
The Hamilton–Jacobi–Bellman equation is the same as in
\eqref{eq:1d-hjb}, except that the drift of the risky asset entering the
Hamiltonian is now the effective drift~\eqref{eq:Strat-S-Ito}.

For logarithmic utility this immediately yields
\[
\theta_t^{*,\,\Strat}
    = \frac{\mu+\tfrac12\sigma^2-r}{\sigma^2}
    = \frac{\mu-r}{\sigma^2} + \frac12
    = \theta_t^{*,\,Ito} + \frac12,
\]
i.e.\ Stratonovich noise leads the log–utility investor to allocate an
additional fraction $1/2$ of wealth into the risky asset.

\subsection{The chain rule viewpoint}

The Stratonovich form has a well--known analytic feature: the classical chain
rule applies to smooth transformations of a \emph{given} state process. Two
caveats are in order about how this feature interacts with the control
problem.

First, the chain rule is a statement about transformations of solutions of
SDEs, not an alternative definition of trading gains. Replacing the
semimartingale gains integral $\int\theta_t a_t\,\dd S_t/S_t$ by a
Stratonovich integral $\int\theta_t a_t\circ\dd S_t/S_t$ would define a
\emph{different} wealth process(the two differ by half the quadratic
covariation between the strategy--wealth process and the price) and such a
definition has no economic justification, since realized trading gains are
limits of forward--looking Riemann sums of actual transactions. For this
reason, gains are defined by the It\^o integral throughout, as stated before
\eqref{eq:wealth-Strat}, and we do not derive the HJB equation from a
Stratonovich wealth equation.

Second, once the wealth process is fixed by \eqref{eq:wealth-Strat}, one may of
course rewrite it, or any smooth function of it, in Stratonovich form and
manipulate the result with classical calculus; this can be convenient for
changes of variables. Such rewritings are representation changes of a fixed
process and have no effect on the optimization.

It follows that \emph{the change in the optimal portfolio rule is not caused
by the chain rule}. The economic effect arises exclusively from the drift
correction $\mu\mapsto\mu+\tfrac12\sigma^2$ induced by holding the nominal
coefficient $\mu$ fixed while adopting the Stratonovich ($\alpha=\tfrac12$)
convention.

\begin{remark}
Under the Stratonovich convention with the nominal drift coefficient $\mu$
held fixed:
\begin{itemize}
    \item the effective It\^o drift of the risky asset increases by $\tfrac12\sigma^2$;
    \item all structural features of the Merton problem remain intact;
    \item the optimal risky allocation increases by exactly~$\tfrac12$;
    \item the chain rule becomes classical, which can simplify transformations
          but does not modify the economics of the control problem.
\end{itemize}
If, instead, the drift coefficient is transformed as
$\mu\mapsto\mu-\tfrac12\sigma^2$ when passing to the Stratonovich convention,
the asset has the same law as the It\^o benchmark and the optimal policy is
unchanged (Proposition~\ref{prop:invariance}). The fixed--coefficient viewpoint
generalizes directly to arbitrary $\alpha\in[0,1]$, where the optimal risky
fraction increases by exactly~$\alpha$.
\end{remark}


\section{\texorpdfstring{$n$ risky assets under the $\alpha$--interpretation}{n risky assets under the alpha-interpretation}}
\label{sec:n-asset-alpha}

We now extend the log–utility consumption–investment problem to a market with
$n$ risky assets.  
Throughout this section $\circ_\alpha$ denotes the stochastic integral in the
$\alpha$–interpretation, and the $n$–dimensional Brownian motion is the one
specified in Section~\ref{sec:math-prelim}.  
We write $\mathbf{1}\in\R^n$ for the vector of ones.

\subsection{\texorpdfstring{Risky-asset dynamics and the $\alpha$–dependent drift shift}{Risky-asset dynamics and the alpha-dependent drift shift}}
Let $S_t=(S_t^1,\dots,S_t^n)^\top\in(0,\infty)^n$ denote the vector of risky
asset prices.  
Fix a constant drift $\mu\in\R^n$ and a constant volatility matrix
$\Sigma\in\R^{n\times n}$, and set
\[
V \;:=\; \Sigma\Sigma^\top.
\]
For notational convenience let $\mathrm{D}(S):=\mathrm{diag}(S^1,\dots,S^n)$.
We model prices in the $\alpha$–interpretation by the vector SDE
\begin{equation}\label{eq:S-alpha}
\dd S_t
  = \mathrm{D}(S_t)\Big(\mu\,\dd t + \Sigma\,\circ_\alpha\dd W_t\Big).
\end{equation}
Componentwise,
\[
\dd S_t^i
  = S_t^i\Bigl(\mu_i\,\dd t + \sum_{k=1}^n \Sigma_{ik}\circ_\alpha\dd W_t^k\Bigr).
\]

\smallskip
We next rewrite the $\alpha$--interpreted SDE in its equivalent It\^o form.
\begin{proposition}\label{prop:alpha-to-ito}
The $\alpha$–SDE \eqref{eq:S-alpha} is equivalent to the Itô SDE
\[
\dd S_t
 = \mathrm{D}(S_t)\Big(\mu^{\mathrm{Ito}}\,\dd t + \Sigma\,\dd W_t\Big),
 \qquad
 \mu^{\mathrm{Ito}}:=\mu+\alpha\,\diag(V),
\]
i.e.\ only the drift changes, while the diffusion remains $\Sigma\,\dd W_t$.
Equivalently, for each $i=1,\dots,n$,
\[
\dd S_t^i
  = S_t^i\Bigl(\mu_i+\alpha V_{ii}\Bigr)\dd t
    + S_t^i\sum_{k=1}^n \Sigma_{ik}\dd W_t^k.
\]
\end{proposition}

\begin{proof}
Write \eqref{eq:S-alpha} as $\dd X_t=a(X_t)\dd t + B(X_t)\circ_\alpha \dd W_t$
with $X=S$, $a(S)=\mathrm{D}(S)\mu$, $B(S)=\mathrm{D}(S)\Sigma$.  
The standard $\alpha$–to–Itô conversion formula gives
\[
a_i^{\mathrm{Ito}}(x)
  = a_i(x)
    + \alpha\sum_{k=1}^n\sum_{j=1}^n B_{jk}(x)\,\partial_{x_j} B_{ik}(x).
\]
Since $B_{ik}(S)=\Sigma_{ik} S^i$ depends only on $S^i$,  
$\partial_{x_j}B_{ik}(S)=\Sigma_{ik}\delta_{ij}$ and hence
\[
\sum_{k,j} B_{jk}(S)\,\partial_{x_j}B_{ik}(S)
  = \sum_k B_{ik}(S)\Sigma_{ik}
  = S^i \sum_k \Sigma_{ik}^2
  = S^i V_{ii}.
\]
Thus $a_i^{\mathrm{Ito}}(S)=S^i(\mu_i+\alpha V_{ii})$ and diffusion remains
$B(S)=\mathrm{D}(S)\Sigma$, yielding the claim.
\end{proof}

In what follows we write
\[
d_V := \diag(V) = (V_{11},\dots,V_{nn})^\top\in\R^n
\]
for the vector of instantaneous variances, so that
$\mu^{\mathrm{Ito}}=\mu+\alpha\,d_V$.

\begin{remark}\label{rem:two-experiments}
Proposition~\ref{prop:alpha-to-ito} supports two logically distinct
comparative exercises. In the \emph{fixed--coefficient} (structural)
experiment, the nominal parameters $(\mu,\Sigma)$ are kept numerically fixed
while $\alpha$ varies; the effective expected returns then become
$\mu^{\mathrm{Ito}}=\mu+\alpha\,d_V$ (in dimension one,
$\mu^{\mathrm{Ito}}=\mu+\alpha\sigma^2$), so varying $\alpha$ changes the law
of returns, the risk premium, and consequently the optimal positions. This is
the case studied in the remainder of the paper. In the
\emph{law--preserving} (representational) experiment, the drift coefficient is
transformed jointly with the convention so that the It\^o law of the asset is
unchanged; then nothing observable, and in particular no optimal policy,
depends on $\alpha$. The next proposition records this formally.
\end{remark}

\smallskip
We now formalize this representation invariance result.
\begin{proposition}\label{prop:invariance}
Fix a target It\^o drift $\bar\mu\in\R^n$ and, for each $\alpha\in[0,1]$, let
the price process solve \eqref{eq:S-alpha} with nominal drift
\[
\mu^{(\alpha)} := \bar\mu - \alpha\,d_V.
\]
Then, for every $\alpha$, the process $S$ is the same It\^o diffusion
\[
\dd S_t=\mathrm{D}(S_t)\bigl(\bar\mu\,\dd t+\Sigma\,\dd W_t\bigr),
\]
so the wealth dynamics, the value function, and the optimal controls of the
consumption--investment problem are independent of $\alpha$:
\[
\theta^*=V^{-1}(\bar\mu-r\mathbf 1),\qquad c_t^*=\rho\,a_t .
\]
More generally, for state--dependent coefficients the same statement holds
with the correction vector $c(x)$ of
Proposition~\ref{prop:alpha-gamma-dictionary} in place of
$\mathrm{D}(S)d_V$. Consequently, any $\alpha$--dependence of optimal policies
reported below is a property of the fixed--coefficient experiment, not of the
stochastic--integral convention per se.
\end{proposition}

\begin{proof}
Immediate from Proposition~\ref{prop:alpha-to-ito}:
$\mu^{(\alpha)}+\alpha\,d_V=\bar\mu$ for every $\alpha$, so the equivalent
It\^o SDE, and with it the entire control problem, does not depend on
$\alpha$. The general statement follows in the same way from
Proposition~\ref{prop:alpha-gamma-dictionary}.
\end{proof}

\subsection{\texorpdfstring{Wealth dynamics under the $\alpha$–interpretation}{Wealth dynamics under the alpha-interpretation}}

Let $\theta_t\in\R^n$ denote the vector of portfolio weights invested in risky
assets, and $c_t\ge0$ denote consumption.  
As in the one-dimensional case, the self–financing identity gives the Itô wealth
dynamics
\begin{equation}\label{eq:wealth-alpha}
\dd a_t
  = a_t\Bigl(r+\theta_t^\top(\mu^{\mathrm{Ito}}-r\mathbf{1})
               - \tfrac{c_t}{a_t}\Bigr)\dd t
    + a_t\,\theta_t^\top\Sigma\,\dd W_t,
\end{equation}
where $\mu^{\mathrm{Ito}}$ is given in Proposition~\ref{prop:alpha-to-ito}.
The only difference from the standard Merton model is the drift adjustment
\(
\mu-r\mathbf{1} \mapsto \mu^{\mathrm{Ito}}-r\mathbf{1}.
\)

\subsection{Optimization problem}

As in the one--asset case, admissible controls $(c_t,\theta_t)$ are required to
be progressively measurable, to satisfy $c_t\ge0$ and $a_t>0$, and to fulfil
the integrability condition~\eqref{eq:int-cond} in its vector form
$\E\int_0^T\theta_t^\top V\theta_t\,\dd t<\infty$, which keeps the
wealth equation well posed and is hypothesis~(i) of
Theorem~\ref{thm:verification} for the whole admissible class.

\medskip
For logarithmic utility, the optimization problem remains time--homogeneous and
Markovian in the wealth variable. We therefore introduce the value function
\begin{equation}\label{eq:J-alpha}
J(a,t)
  := \sup_{(c,\theta)}
      \E\!\left[
         \int_t^\infty e^{-\rho s}\ln c_s\,\dd s\;\Big|\;a_t=a
      \right],
\end{equation}
with calendar--time discounting as in \eqref{eq:1d-objective}.
Standard dynamic programming arguments then yield the Hamilton--Jacobi--Bellman
equation
\begin{equation}\label{eq:HJB-alpha}
0=\sup_{c>0,\ \theta\in\R^n}\Big\{
     e^{-\rho t}\ln c
     + J_t
     +J_a\bigl(a(r+\theta^\top(\mu^{\mathrm{Ito}}-r\mathbf{1}))-c\bigr)
     +\tfrac12 J_{aa}\,a^2\,\theta^\top V \theta
  \Big\}.
\end{equation}
Compared to the classical Merton HJB, the only modification is the replacement
of the excess--return vector $\mu-r\mathbf{1}$ by its $\alpha$--adjusted version
$\mu^{\mathrm{Ito}}-r\mathbf{1}$, as identified in
Proposition~\ref{prop:alpha-to-ito}.

\subsection{Optimal policies and value function}

The logarithmic utility function preserves the homothetic structure of the
problem even in the presence of multiple risky assets. As a consequence, the
HJB equation admits an explicit solution, and optimal policies can be obtained
in closed form.

\smallskip
We now state the solution for log utility with $n$ risky assets under an $\alpha$--interpreted model.
\begin{theorem}\label{thm:n-asset-alpha}
Assume that the covariance matrix $V=\Sigma\Sigma^\top$ is positive definite.
Then the optimal consumption rate and portfolio weights are given by
\begin{equation}\label{eq:opt-alpha}
c_t^* = \rho\,a_t,
\qquad
\theta_\alpha^*
  = V^{-1}(\mu^{\mathrm{Ito}}-r\mathbf{1})
  = V^{-1}(\mu-r\mathbf{1})
    + \alpha\,V^{-1}d_V.
\end{equation}
The corresponding value function is
\[
J(a,t)
  = \Bigl(\beta_0+\rho^{-1}\ln a\Bigr)e^{-\rho t},
\qquad
\beta_0
  = \frac{1}{\rho}\Bigl(\tfrac{r}{\rho}+\ln\rho-1\Bigr)
    +\frac{1}{2\rho^2}(\mu^{\mathrm{Ito}}-r\mathbf{1})^\top
       V^{-1}(\mu^{\mathrm{Ito}}-r\mathbf{1}).
\]
\end{theorem}

\begin{remark}[Direction of the portfolio adjustment]\label{rem:sensitivity-direction}
Under the fixed--coefficient experiment, the $\alpha$--interpretation modifies
optimal portfolio choice exclusively through an additive shift in the vector
of excess returns, with sensitivity
\[
\frac{\partial\theta_\alpha^*}{\partial\alpha}
  = V^{-1}d_V.
\]
This vector equals $\mathbf 1$, so that the correction increases every
position by exactly the same amount $\alpha$, as in the naive reading of the
univariate rule, precisely when $d_V=V\mathbf 1$,
i.e.\ when every off--diagonal row sum of $V$ vanishes: writing $V=D+E$ with
$D=\diag(d_V)$ and $E$ the off--diagonal part, the condition is
$E\mathbf 1=\mathbf 0$. For $n\le3$ this linear restriction forces $E=0$, so
in low dimension $V^{-1}d_V=\mathbf 1$ holds if and only if returns are
instantaneously uncorrelated. From $n=4$ onwards, however, non--diagonal
solutions exist: for example,
\[
V=\begin{pmatrix}4&1&-1&0\\ 1&4&0&-1\\ -1&0&4&1\\ 0&-1&1&4\end{pmatrix}
\]
is positive definite (eigenvalues $2,4,4,6$) and satisfies
$V^{-1}d_V=\mathbf 1$ although its returns are correlated. The exact
condition for the ``$+\alpha$ on every asset'' reading is thus
$E\mathbf 1=\mathbf 0$, a knife--edge restriction on the covariance matrix
that we state directly rather than through a genericity gloss.

Away from this restriction the correction is in general no longer
$\alpha$ per asset. Two failures should be distinguished. First, the shift
can be \emph{non--uniform and even negative}: with
\[
V=\begin{pmatrix}1 & 1.5\\ 1.5 & 4\end{pmatrix},
\qquad
d_V=\begin{pmatrix}1\\ 4\end{pmatrix},
\qquad
V^{-1}d_V=\frac{1}{1.75}\begin{pmatrix}-2\\ 2.5\end{pmatrix}
\approx\begin{pmatrix}-1.14\\ 1.43\end{pmatrix},
\]
the position in the first asset \emph{decreases} as $\alpha$ increases, even
though $V$ is positive definite. Second, the components of $V^{-1}d_V$ may all
coincide at a common value $\neq1$: this happens precisely when
$d_V\in\operatorname{span}(V\mathbf 1)$, a strictly weaker condition than
$E\mathbf 1=\mathbf 0$. In an equal--variance market $d_V=v_0\mathbf 1$, so
this reduces to $V\mathbf 1\propto\mathbf 1$, i.e.\ to the covariance matrix
having constant row sums. This property is automatic for a two--asset
equal--variance covariance matrix,every symmetric $2\times2$ matrix with
constant diagonal has equal row sums  but it need not hold in dimensions
$n\ge3$. For instance,
\[
V=\begin{pmatrix}1 & 1/5 & 1/10\\ 1/5 & 1 & 3/10\\ 1/10 & 3/10 & 1\end{pmatrix}
\]
has equal unit variances and is positive definite (its leading principal minors
are $1$, $24/25$, and $\det V=109/125>0$), yet its row sums
$V\mathbf 1=(13/10,\,3/2,\,7/5)^\top$ are not constant, so
\[
V^{-1}d_V=V^{-1}\mathbf 1
=\tfrac{1}{218}\begin{pmatrix}175\\ 135\\ 160\end{pmatrix}
\approx\begin{pmatrix}0.803\\ 0.619\\ 0.734\end{pmatrix}
\]
is not uniform. Equal variances alone therefore do not guarantee a uniform
correction. So a \emph{uniform} shift does not by itself signal
$E\mathbf 1=\mathbf 0$, and only the sharper reading
``$+\alpha$ on every asset'' carries that characterization. What is robust,
and is all the sequel uses, is that under correlation the per--asset
correction generally departs from $\alpha$ and can reverse sign. When $n=1$
and $V=\sigma^2$, the rule reduces to the transparent formula
\[
\theta^*_\alpha
  = \frac{\mu-r}{\sigma^2} + \alpha,
\]
showing that the optimal risky position moves one--for--one with $\alpha$: over
the full range $\alpha\in[0,1]$ from the It\^o to the Klimontovich convention,
exposure rises by exactly one unit of portfolio weight.
\end{remark}

To gauge the economic size of this shift, Figure~\ref{fig:theta-alpha} plots
$\theta^*_\alpha$ against $\alpha\in[0,1]$ for three stylized,
\emph{independent} single--asset markets spanning realistic combinations of
excess return and volatility (three separate univariate problems, not three
assets held in one portfolio). Because the correction is exactly $+\alpha$
regardless of $(\mu-r,\sigma)$, the three curves are parallel with unit slope;
what differs across markets is the baseline $\theta_0=(\mu-r)/\sigma^2$
against which this shift is measured. Moving from the It\^o to the
Klimontovich convention adds a full unit of portfolio weight in every case,
but this represents a $48\%$ increase in risky exposure for the
low--volatility asset versus a $153\%$ increase for the high--volatility one:
the same absolute correction is economically far more consequential whenever
the baseline myopic demand $\theta_0$ is small. Two caveats apply. First,
baseline positions such as $\theta_0\approx2.08$ entail substantial leverage;
no borrowing or short--sale constraints are imposed anywhere in the paper.
Second, the figure depicts the fixed--coefficient experiment: under the
law--preserving experiment of Proposition~\ref{prop:invariance}, all three
curves would be horizontal.

\begin{figure}[htbp]
\centering
\includegraphics[width=0.75\linewidth]{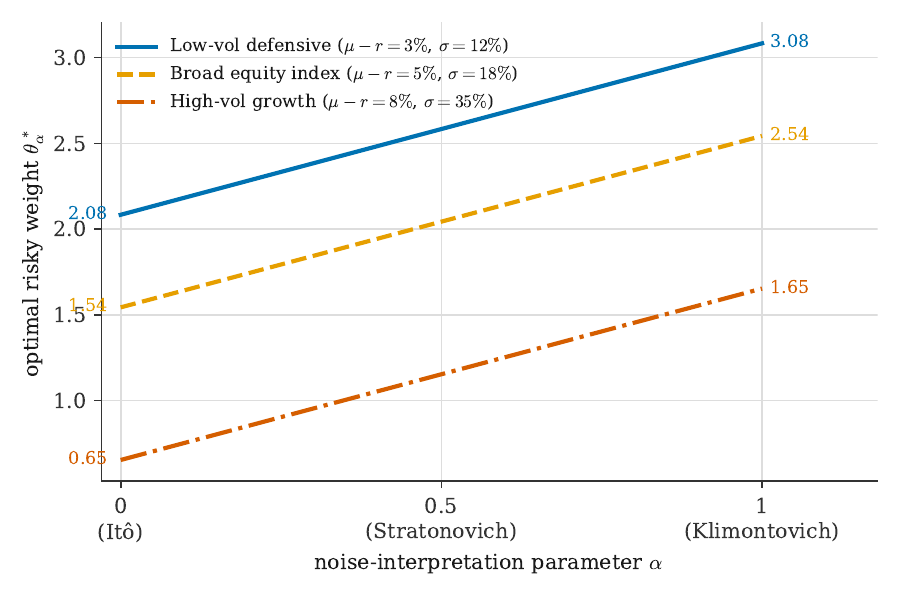}
\caption{Optimal risky weight $\theta^*_\alpha=(\mu-r)/\sigma^2+\alpha$ as a
function of the noise--interpretation parameter $\alpha$, for three
representative independent single--asset markets, under the
fixed--coefficient experiment (nominal $\mu$ held fixed across conventions).
Labels report $\theta^*_\alpha$ at $\alpha=0$ (It\^o) and $\alpha=1$
(Klimontovich).}
\label{fig:theta-alpha}
\end{figure}

\begin{proof}[Proof of Theorem~\ref{thm:n-asset-alpha}]
Let $\lambda:=\mu^{\mathrm{Ito}}-r\mathbf{1}$. Motivated by the homotheticity of
logarithmic utility, we insert the ansatz
\(
J(a,t)=\bigl(\beta_0+\beta_1\ln a\bigr)e^{-\rho t}
\)
into the HJB equation~\eqref{eq:HJB-alpha}. With
$J_t=-\rho(\beta_0+\beta_1\ln a)e^{-\rho t}$, $J_a=\beta_1 a^{-1}e^{-\rho t}$,
and $J_{aa}=-\beta_1 a^{-2}e^{-\rho t}$, dividing out the common factor
$e^{-\rho t}$ reduces \eqref{eq:HJB-alpha} to
\[
0=\sup_{c>0,\ \theta\in\R^n}
\Bigl\{\ln c-\rho\beta_0-\rho\beta_1\ln a
  - \tfrac{\beta_1}{a}c + \beta_1 r + \beta_1\theta^\top\lambda
    - \tfrac12\beta_1\theta^\top V\theta\Bigr\}.
\]
The expression in braces is strictly concave in $(c,\theta)$, and the
first--order conditions therefore identify the unique maximizers
\[
c^*(a)=\frac{a}{\beta_1},
\qquad
\theta^*=V^{-1}\lambda.
\]
Substituting these expressions back, the coefficient of $\ln a$ is
$1-\rho\beta_1$, which must vanish, yielding $\beta_1=1/\rho$ and hence
$c^*=\rho a$; the remaining constant terms determine $\beta_0$ uniquely as
stated. Positivity of $c^*$ follows from $\beta_1>0$, and uniqueness
from strict concavity. Finally, optimality (not merely stationarity of the HJB)
follows from Theorem~\ref{thm:verification}: with the factor absent,
$\mu_{\mathrm{eff}}=\mu^{\mathrm{Ito}}$ and $V=\Sigma\Sigma^\top$,
hypothesis~(i) holds for every admissible control by the vector form of
Definition~\ref{def:admissible}(iii), hypothesis~(ii) holds trivially since
$v$ is constant, and
Remark~\ref{rem:verification-trivial} shows that admissibility of
$(c^*,\theta^*)$ (here $\theta^*$ is constant) and the transversality
condition hold because wealth is again a
geometric Brownian motion under this control, so Theorem~\ref{thm:verification}
applies directly.
\end{proof}

\section{\texorpdfstring{Factor--driven risky asset with correlated noise under the $\alpha$--interpretation}{Factor-driven risky asset with correlated noise under the alpha-interpretation}}
\label{sec:factor-correlated-noise}

The analysis of Sections~\ref{sec:n-asset-alpha} shows that, in markets with
constant volatilities, the $\alpha$--interpretation can be absorbed into a
deterministic modification of expected returns. From a structural viewpoint,
this raises a natural question: under which modelling assumptions does the
interpretation of stochastic integration have genuinely state--dependent
consequences?

In this section we address this question by introducing a factor--driven
volatility, coupled to the risky return through correlated Brownian noise. This
setting encompasses classical stochastic volatility models and reveals that the
$\alpha$--dependent drift adjustment of the return decomposes into \emph{two}
channels: an own--variance component $\alpha\,\sigma(X_t)^2$, which becomes
state--dependent as soon as the volatility varies with the factor, and a
covariation component $\alpha\,\varrho\,\sigma'(X_t)\nu(X_t)$ generated by the
instantaneous covariation between the factor and the return noise. The
resulting corrections to optimal portfolio choice are no longer static, but
depend on the current state of the factor.

\subsection{Financial market}\label{subsec:factor-market}

We now specify the financial market underlying the factor--driven model. The distinguishing
feature of this setting is the presence of an auxiliary factor process whose
fluctuations affect the volatility of the risky asset and are correlated with
the return noise.

\medskip
We work on a filtered probability space
$(\Omega,\mathcal F,(\mathcal F_t)_{t\ge0},\mathbb P)$ supporting a
two--dimensional Brownian motion
\[
W_t := (W_t^{S},W_t^{X})^\top,
\qquad t\ge 0,
\]
whose components drive, respectively, the risky return and the factor dynamics.
The instantaneous covariance structure is given by
\begin{equation}\label{eq:cov-matrix}
    [W]_t = R\,t,
    \qquad
    R :=
    \begin{pmatrix}
    1 & \varrho\\
    \varrho & 1
    \end{pmatrix},
    \qquad \varrho\in[-1,1].
\end{equation}
Equivalently,
\[
\dd\langle W^{S},W^{X}\rangle_t = \varrho\,\dd t,
\qquad
\dd\langle W^{S}\rangle_t=\dd\langle W^{X}\rangle_t=\dd t.
\]
The correlation parameter $\varrho$ quantifies the instantaneous coupling
between factor fluctuations and asset returns and will play a central role in
the $\alpha$--dependent corrections derived below.

\medskip
As before, the money--market account $b=(b_t)_{t\ge0}$ evolves according to
\eqref{eq:money-market}. The risky asset $S=(S_t)_{t\ge0}$ is modeled in the
$\alpha$--interpretation by
\begin{equation}\label{eq:S-factor-alpha-4}
    \frac{\dd S_t}{S_t}
    =
    \mu(X_t)\,\dd t
    +
    \sigma(X_t)\,\circ_\alpha \dd W_t^{S},
\end{equation}
where the drift and volatility depend on the current state of the factor.
Throughout we assume $\mu,\sigma\in C^1(\mathbb R)$ and $\sigma(x)\neq 0$ for all
$x$.

\medskip
The factor process $X=(X_t)_{t\ge0}$ evolves autonomously according to
\begin{equation}\label{eq:X-factor-alpha-4}
    \dd X_t
    =
    b(X_t)\,\dd t
    +
    \nu(X_t)\,\circ_\alpha \dd W_t^{X},
\end{equation}
with $b,\nu\in C^1(\mathbb R)$. While $X$ does not depend on $S$ directly, the
shared noise structure encoded in~\eqref{eq:cov-matrix} induces an indirect
coupling that will become apparent after conversion to It\^o form.

\medskip
Unless stated otherwise, coefficients are chosen so that
\eqref{eq:S-factor-alpha-4}--\eqref{eq:X-factor-alpha-4} admit unique strong
solutions, and controls are understood to be admissible in the sense of
Definition~\ref{def:admissible}, complemented for the candidate policies by
Assumption~\ref{ass:factor}(iii). In addition, the present section allows for
state dynamics that fall outside the globally Lipschitz framework adopted
earlier, as formalized below.

\begin{assumption}\label{ass:factor}
In Section~\ref{sec:factor-correlated-noise} we allow state dynamics with
coefficients that are not globally Lipschitz (e.g.\ square--root diffusions),
as is standard in stochastic volatility models.

\begin{enumerate}[(i)]
\item \textbf{Local well--posedness.}
The factor/asset system (in It\^o form) admits a unique strong solution up to
its maximal lifetime $\tau\in(0,\infty]$.

\item \textbf{Nonattainment of boundaries / positivity (when relevant).}
In models with a boundary (e.g.\ $V_t\ge 0$), parameters are chosen so that the
boundary is not attained, hence $\tau=\infty$ and the state remains in its
natural domain almost surely.

\item \textbf{Admissibility of feedback controls.}
The candidate feedback controls derived in this section are admissible in the
sense of Definition~\ref{def:admissible}; in particular the corresponding
wealth process remains strictly positive and the factor--driven integrability
condition
\[
\mathbb E\!\left[\int_0^T \pi_t^2\sigma(X_t)^2\,\dd t\right]<\infty
\qquad(T>0)
\]
holds, together with the terminal condition~\eqref{eq:adm-terminal}.
\end{enumerate}
\end{assumption}

\subsection{Conversion to It\^o form}\label{subsec:factor-Ito}

Equation \eqref{eq:X-factor-alpha-4} is equivalent to the It\^o SDE
\begin{equation}\label{eq:X-factor-Ito-4}
    \dd X_t
    =
    \Big(b(X_t) + \alpha\,\nu(X_t)\nu'(X_t)\Big)\dd t
    +
    \nu(X_t)\,\dd W_t^{X};
\end{equation}
we write $b_\alpha(x):=b(x)+\alpha\,\nu(x)\nu'(x)$ for the converted factor
drift, which is the drift entering every generator--based computation below.

For the risky asset, the conversion must be carried out for the \emph{joint}
state $(S,X)$, because the diffusion coefficient of $S$, namely
$S_t\sigma(X_t)$, depends both on the price itself and on the factor. To apply
the general conversion formula of
Proposition~\ref{prop:alpha-gamma-dictionary}, which is stated for independent
driving Brownian motions, write
\[
W^{S}=B^1,\qquad W^{X}=\varrho\,B^1+\sqrt{1-\varrho^2}\,B^2
\]
with $B=(B^1,B^2)$ a standard two--dimensional Brownian motion (see
Remark~\ref{rem:correlated-appendix}). The state $(S,X)$ then has diffusion
matrix
\[
G(S,x)
=
\begin{pmatrix}
S\,\sigma(x) & 0\\[2pt]
\varrho\,\nu(x) & \sqrt{1-\varrho^2}\,\nu(x)
\end{pmatrix},
\]
and the $\alpha$--to--It\^o drift correction for the $S$--component is
\[
c_S(S,x)
=\sum_{j\in\{S,x\}}\sum_{k=1}^{2}G_{jk}(S,x)\,\partial_j G_{Sk}(S,x),
\]
whose two nonzero contributions are
\[
\underbrace{S\sigma(x)\cdot\sigma(x)}_{G_{S1}\,\partial_S G_{S1}}
\qquad\text{and}\qquad
\underbrace{\varrho\,\nu(x)\cdot S\sigma'(x)}_{G_{X1}\,\partial_x G_{S1}},
\]
so that
\[
c_S(S,x)= S\bigl(\sigma(x)^2+\varrho\,\sigma'(x)\nu(x)\bigr).
\]
The first term is the \emph{own--variance} channel, generated by the
dependence of the diffusion coefficient on the price itself (exactly as in the
constant--volatility case, where it produces the correction $\alpha\sigma^2$);
the second is the \emph{covariation} channel, generated by the dependence of
the volatility on the correlated factor. The same computation for the
$X$--component returns
$c_X=\varrho^2\nu\nu'+(1-\varrho^2)\nu\nu'=\nu\nu'$, consistent
with~\eqref{eq:X-factor-Ito-4}. Therefore \eqref{eq:S-factor-alpha-4} is
equivalent to the It\^o SDE
\begin{equation}\label{eq:S-factor-Ito-4}
    \frac{\dd S_t}{S_t}
    =
    \Big(\mu(X_t) + \alpha\bigl(\sigma(X_t)^2+\varrho\,\sigma'(X_t)\nu(X_t)\bigr)\Big)\dd t
    +
    \sigma(X_t)\,\dd W_t^{S}.
\end{equation}

\begin{remark}[Two channels for the $\alpha$--correction]\label{rem:alpha-trivial}
The $\alpha$--dependent drift correction in \eqref{eq:S-factor-Ito-4} is the sum
of an own--variance term $\alpha\,\sigma(X_t)^2$ and a covariation term
$\alpha\,\varrho\,\sigma'(X_t)\nu(X_t)$. The correlation $\varrho$ governs only
the second term: if $\varrho=0$ (independent factor and return shocks), the
covariation channel closes, but the own--variance correction
$\alpha\,\sigma(X_t)^2$ remains, and it is state--dependent whenever $\sigma$
is nonconstant. Conversely, if $\sigma$ is constant the covariation term
vanishes ($\sigma'\equiv0$) and the own--variance term reduces to the constant
shift $\alpha\sigma^2$ of Section~\ref{sec:n-asset-alpha}. Thus, under the
fixed--coefficient experiment, the $\alpha$--interpretation always affects the
effective drift; correlation determines only whether an \emph{additional},
factor--dependent channel is present.
\end{remark}

\subsection{Wealth dynamics}\label{subsec:wealth-factor}

We next derive the wealth dynamics induced by the factor--dependent asset
returns. Let $A_t$ denote the investor’s wealth at time $t$, and let $\pi_t$
denote the fraction of wealth invested in the risky asset, with the remainder
invested in the money--market account. As before, $c_t$ denotes the consumption
rate.

Under the self--financing constraint, and using the It\^o form of the asset
dynamics derived above, the wealth process satisfies
\begin{equation}\label{eq:wealth-factor-4}
    \frac{\dd A_t}{A_t}
    =
    \Big(
        r
        +
        \pi_t\big(\mu_{\mathrm{eff}}(X_t)-r\big)
        -
        \frac{c_t}{A_t}
    \Big)\dd t
    +
    \pi_t\sigma(X_t)\,\dd W_t^{S},
\end{equation}
where the \emph{effective return drift} is given by
\begin{equation}\label{eq:mu-eff-factor-4}
    \mu_{\mathrm{eff}}(x)
    :=
    \mu(x) + \alpha\bigl(\sigma(x)^2+\varrho\,\sigma'(x)\nu(x)\bigr).
\end{equation}

Relative to the constant--volatility case, the structure of the wealth equation
is unchanged. All effects of the $\alpha$--interpretation are captured by the
drift correction
$\mu_{\mathrm{eff}}(X_t)-\mu(X_t)
=\alpha\sigma(X_t)^2+\alpha\varrho\,\sigma'(X_t)\nu(X_t)$,
which is state--dependent through both channels. When $\varrho=0$ the
covariation part disappears, but the own--variance part
$\alpha\sigma(X_t)^2$ persists, so the $\alpha$--interpretation continues to
affect the wealth dynamics under the fixed--coefficient experiment.

\subsection{Logarithmic utility maximization}\label{subsec:log-factor}

We now consider the investor’s optimal consumption--investment problem in the
factor--driven market. The objective is to maximize discounted expected utility
of consumption,
\[
    \sup_{(\pi,c)}
    \mathbb E\Bigg[
        \int_0^\infty e^{-\rho t}\log(c_t)\,\dd t
    \Bigg],
    \qquad \rho>0,
\]
over admissible controls $(\pi_t,c_t)$.

Despite the presence of a stochastic factor, the problem retains a simple
structure under logarithmic utility. In particular, log utility eliminates
intertemporal hedging motives and preserves homotheticity in wealth. Motivated
by this observation, we seek a value function of the separable form
\[
V(a,x)=\rho^{-1}\log a + v(x),
\]
where $v$ accounts for the contribution of the factor process; this is the
ansatz $\phi(a,x)=\rho^{-1}\ln a+v(x)$ of Theorem~\ref{thm:verification}, and
the same normalization used in Sections~\ref{subsec:merton-classical}
and~\ref{sec:n-asset-alpha}.

Substituting this ansatz into the associated Hamilton--Jacobi--Bellman equation
and optimizing pointwise with respect to $c$ and $\pi$ yields the first--order
conditions
\begin{equation}\label{eq:controls-factor-4}
    c_t^\ast = \rho A_t,
    \qquad
    \pi_t^\ast
    =
    \frac{\mu_{\mathrm{eff}}(X_t)-r}{\sigma(X_t)^2}.
\end{equation}
Using the explicit form of the effective drift \eqref{eq:mu-eff-factor-4}, the
candidate risky fraction can be decomposed as
\begin{equation}\label{eq:pi-alpha-factor-4}
    \pi_t^\ast
    =
    \frac{\mu(X_t)-r}{\sigma(X_t)^2}
    +
    \alpha
    +
    \alpha\,\varrho\,\frac{\sigma'(X_t)\nu(X_t)}{\sigma(X_t)^2}.
\end{equation}
The first term coincides with the classical myopic demand evaluated at the
current factor level. The second term, $+\alpha$, is the own--variance
correction --- the exact analogue of the constant shift found in the
constant--volatility case, and present regardless of the correlation. The
third term is the covariation correction, originating from the interaction
between the factor dynamics and the correlated return noise.

The covariation term vanishes in each of the following cases: (i) the
volatility coefficient $\sigma$ is constant, so that factor fluctuations do
not affect the return variance; (ii) the factor has no diffusive component
($\nu\equiv0$); or (iii) the factor and return noises are uncorrelated
($\varrho=0$). In all three situations, the state--dependent part of the
$\alpha$--correction disappears and the portfolio rule reduces to the
classical myopic demand shifted by the constant $\alpha$ --- the behaviour of
Section~\ref{sec:n-asset-alpha}. Under the fixed--coefficient experiment the
$\alpha$--interpretation therefore never becomes irrelevant; under the
law--preserving experiment of Proposition~\ref{prop:invariance} it is
irrelevant by construction.

\smallskip
We next provide a verification result for the factor--driven model.
\begin{proposition}\label{prop:factor-verification}
The controls~\eqref{eq:controls-factor-4} are optimal in the sense of
Theorem~\ref{thm:verification} (with $n=1$, $\Sigma(x)=\sigma(x)$, and the value
function $V(a,x)=\rho^{-1}\ln a+v(x)$ of this subsection playing the role of
$\phi$ in Theorem~\ref{thm:verification}) provided $v\in C^2$ solves the linear
ODE
\begin{equation}\label{eq:v-ode}
    \tfrac12\nu(x)^2v''(x)+b_\alpha(x)v'(x)-\rho v(x) = -g(x),
    \qquad
    g(x):=\ln\rho-1+\frac r\rho+\frac{1}{2\rho}\left(\frac{\mu_{\mathrm{eff}}(x)-r}{\sigma(x)}\right)^{\!2},
\end{equation}
obtained by specializing~\eqref{eq:verification-ode} to a scalar factor.
Note that the drift entering \eqref{eq:v-ode} is the \emph{converted} factor
drift $b_\alpha(x)=b(x)+\alpha\nu(x)\nu'(x)$ from \eqref{eq:X-factor-Ito-4},
since the generator of the factor is computed from its It\^o form (in
Theorem~\ref{thm:verification} the factor dynamics are already stated in
It\^o form). Optimality further requires the remaining hypotheses of
Theorem~\ref{thm:verification} to hold in this setting: the martingale
condition~\eqref{eq:verification-martingale} for every admissible control
(built into Definition~\ref{def:admissible}(iii) in its factor--driven
form), the integrability condition~\eqref{eq:verification-v-integrability}
for $v$ along the factor process $x_t$, the liminf
condition~\eqref{eq:verification-liminf} for every admissible control, and
admissibility of the candidate~\eqref{eq:controls-factor-4} together with
the transversality condition~\eqref{eq:verification-transversality} along
the resulting wealth and factor processes.
Equation~\eqref{eq:v-ode} is the resolvent
equation of the generator of $x$, so whenever
$\E_x\bigl[\int_0^\infty e^{-\rho s}|g(x_s)|\,\dd s\bigr]<\infty$ for every $x$ in
the state space, the Feynman--Kac functional
\begin{equation}\label{eq:v-feynman-kac}
    v(x) = \E_x\left[\int_0^\infty e^{-\rho s}g(x_s)\,\dd s\right]
\end{equation}
is well defined and is the natural candidate for $v$. Whether this candidate is
a classical $C^2$ solution --- and hence discharges the hypothesis $v\in C^2$ of
Theorem~\ref{thm:verification} rather than merely representing it --- depends on
the regularity of the resolvent of the factor generator. For a regular
one--dimensional factor diffusion with inaccessible boundaries and integrable
$g$, this is supplied by the classical Green--function theory of the
$\rho$--resolvent \citep{ItoMcKean1974,BorodinSalminen2002}: the Feynman--Kac
functional is then $C^2$ on the interior of the state space and solves
\eqref{eq:v-ode} classically. We carry this out in full for the Heston
square--root factor in Appendix~\ref{app:heston-verification}, where $v\in C^2$,
together with the remaining hypotheses of Theorem~\ref{thm:verification}, is
established under the strict Feller condition~\eqref{eq:cir-inverse-moment}. For
a general factor the statement of this proposition remains conditional on
$v\in C^2$.
\end{proposition}

Because $V(a,x)=\rho^{-1}\ln a+v(x)$ is additively separable, the candidate
controls~\eqref{eq:controls-factor-4} do not depend on $v$ at all — this is the
familiar myopic property of logarithmic utility, which eliminates intertemporal
hedging demand. The function $v$ enters only through the value \emph{level}, via
\eqref{eq:v-feynman-kac}, so its existence is immaterial to the algebraic
\emph{form} of the candidate policies reported throughout this section; it is,
however, required to \emph{certify} their optimality, since
Theorem~\ref{thm:verification} presupposes a $C^2$ solution
of~\eqref{eq:v-ode}, without it, \eqref{eq:controls-factor-4} identifies
candidates rather than verified optima in the general factor model. The integrability condition
in~\eqref{eq:v-feynman-kac} is exactly the requirement that the discounted
squared Sharpe ratio of the factor--driven asset be integrable; for the Heston
specialization of Section~\ref{subsec:heston-example}, this reduces to an
explicit condition on the variance process, discussed there, under which the
certification is completed in Appendix~\ref{app:heston-verification}.

\subsection{Example: Heston stochastic volatility}\label{subsec:heston-example}

We illustrate the mechanism identified above by specializing to the Heston
stochastic volatility model, which provides a canonical example of a
factor--driven risky asset with correlated noise. In this setting, the factor
process represents the instantaneous variance of returns, and correlation
between return and variance shocks is an empirically well--established feature
of financial markets. Under the classical It\^o convention ($\alpha=0$), this
setting belongs to the stochastic investment opportunity set literature
\citep{KimOmberg1996,Liu2007,Kraft2005,ChackoViceira2005}, with
\citet{Kraft2005} (power utility) and \citet{ChackoViceira2005} (recursive
preferences in incomplete markets) treating Heston--type specifications. For
logarithmic utility, the $\alpha=0$ policy below is the classical myopic rule
$(\mu-r)/V_t$, which is the log--utility benchmark nested in that literature;
the results below isolate the additional effect of the noise interpretation
on top of this benchmark.

We take the factor to be the variance process $V_t$ and set $X_t=V_t$, with
coefficients
\[
\sigma(v)=\sqrt v,\qquad
b(v)=\kappa(\theta - v),\qquad
\nu(v)=\xi\sqrt v,
\]
where $\kappa,\theta,\xi>0$ and $V_0>0$. Since $\sigma(v)=\sqrt v$ is $C^1$
on $(0,\infty)$ but not at the boundary $v=0$, all conversions in this
example are performed on the open domain $(0,\infty)$; this is sufficient
because the boundary is not attained under the strict Feller condition
adopted below (Assumption~\ref{ass:factor}(ii)). The resulting
$\alpha$--interpreted dynamics are
\begin{align}
    \frac{\dd S_t}{S_t}
    &= \mu\,\dd t + \sqrt{V_t}\,\circ_\alpha \dd W_t^{S},
    \label{eq:heston-alpha-S}\\
    \dd V_t
    &= \kappa(\theta - V_t)\,\dd t + \xi\sqrt{V_t}\,\circ_\alpha \dd W_t^{X},
    \label{eq:heston-alpha-V}
\end{align}
with $\dd\langle W^{S},W^{X}\rangle_t=\varrho\,\dd t$.

In the Heston case the $\alpha$--dependent corrections can be computed explicitly.
Since $\sigma(v)^2=v$, and since $\sigma'(v)=\frac{1}{2\sqrt v}$ and
$\nu(v)=\xi\sqrt v$ give $\sigma'(v)\nu(v)=\xi/2$, the own--variance channel
contributes the state--dependent term $\alpha V_t$ while the covariation
channel contributes the constant $\alpha\varrho\xi/2$. The It\^o form of
\eqref{eq:heston-alpha-S} is therefore
\begin{equation}\label{eq:heston-Ito-S}
    \frac{\dd S_t}{S_t}
    =
    \Big(\mu + \alpha\Big(V_t+\varrho\,\frac{\xi}{2}\Big)\Big)\dd t
    +
    \sqrt{V_t}\,\dd W_t^{S}.
\end{equation}
Thus the effective expected return combines a constant covariation shift with
an own--variance shift proportional to the current variance level.

For the variance process, the identity
$\nu(v)\nu'(v)=\xi^2/2$ shows that \eqref{eq:heston-alpha-V} is equivalent to
\begin{equation}\label{eq:heston-Ito-V}
    \dd V_t
    =
    \Big(\kappa(\theta - V_t) + \alpha\,\frac{\xi^2}{2}\Big)\dd t
    +
    \xi\sqrt{V_t}\,\dd W_t^{X}.
\end{equation}
Equivalently, the drift can be rewritten as $\kappa(\theta_\alpha - V_t)$ with
\begin{equation}\label{eq:theta-alpha}
    \theta_\alpha := \theta + \alpha\,\frac{\xi^2}{2\kappa}.
\end{equation}
The $\alpha$--interpretation therefore shifts the long--run mean of the variance
process while preserving its CIR structure.

In this model the effective return drift \eqref{eq:mu-eff-factor-4} becomes
\[
\mu_{\mathrm{eff}}(v)
=
\mu + \alpha\Big(v+\varrho\,\frac{\xi}{2}\Big).
\]
Consequently, the optimal consumption remains proportional to wealth, while
the optimal risky fraction takes the explicit form
\begin{equation}\label{eq:heston-controls}
    c_t^\ast = \rho A_t,
    \qquad
    \pi_t^\ast
    =
    \frac{\mu_{\mathrm{eff}}(V_t)-r}{V_t}
    =
    \alpha + \frac{K_\alpha}{V_t},
    \qquad
    K_\alpha := \mu-r+\alpha\,\varrho\,\frac{\xi}{2}.
\end{equation}
The policy is thus the sum of a constant own--variance shift $\alpha$ and an
inverse--variance component $K_\alpha/V_t$ whose coefficient combines the
excess return with the covariation correction. The state--dependent part is
amplified during low--volatility periods and fades when volatility is high,
while the constant shift $+\alpha$ dominates in high--volatility states.
Relative to the classical log--utility myopic rule $(\mu-r)/V_t$, recovered at
$\alpha=0$, the $\alpha$--interpretation enters as the additive perturbation
$\alpha+\alpha\varrho\xi/(2V_t)$, without altering the functional form of the
policy.

Under the strict Feller condition~\eqref{eq:cir-inverse-moment} adopted below,
$(c^\ast,\pi^\ast)$ are the \emph{optimal} controls, not merely candidates:
they solve the pointwise first--order conditions of
Theorem~\ref{thm:verification}, and all of that theorem's hypotheses are
discharged for the Heston state in
Appendix~\ref{app:heston-verification}. Specifically, the integrability
condition~\eqref{eq:verification-martingale} on $\pi^\ast$
(Definition~\ref{def:admissible}(iii)) and the integrability of $g$ underlying
\eqref{eq:v-feynman-kac} are both implied by the CIR inverse--moment bound
(see~\eqref{eq:cir-inverse-moment}). When $K_\alpha\neq0$, this inverse moment
is also necessary because the relevant expressions contain a
$K_\alpha^2/V_t$ term; when $K_\alpha=0$, that term disappears and the first
moment of $V_t$ suffices. The classical Green--function theory of the
$\rho$--resolvent \citep{ItoMcKean1974,BorodinSalminen2002} then shows that
the Feynman--Kac functional is a genuine $C^2$ solution
of~\eqref{eq:v-ode}; and the same bound
closes the factor--integrability hypothesis~\eqref{eq:verification-v-integrability}
as well as the liminf and transversality
conditions~\eqref{eq:verification-liminf}--\eqref{eq:verification-transversality}
(the latter being the terminal part of full admissibility). The
$\alpha$--comparison, moreover, does not depend on the verification argument at
all: the sensitivity of the policy to the convention is a statement about the
first--order conditions, which hold irrespective of it.

Unlike the constant--volatility case of Section~\ref{sec:n-asset-alpha}, here
the two components of the correction can pull in opposite directions. Since
$\varrho<0$ is the empirically standard case for equity indices (the leverage
effect), the covariation part reduces the inverse--variance demand while the
own--variance part adds a positive constant. Figure~\ref{fig:heston-pi}
illustrates the outcome with a calibration common in the
stochastic--volatility literature, $\kappa=2$, $\theta=4.5\%$, $\xi=0.4$,
$\varrho=-0.5$, and $\mu-r=6\%$ (the strict Feller condition
$2\kappa\theta_\alpha>\xi^2$ of \eqref{eq:theta-alpha} holds for every
$\alpha\in[0,1]$ under this calibration; see the discussion following
\eqref{eq:cir-inverse-moment} below for the sense in which this makes
$\pi_t^\ast$ admissible). Under the It\^o convention ($K_0=0.06$) the investor
holds a leveraged long position that intensifies as volatility falls. Under
the Stratonovich convention ($K_{1/2}=0.01$) the inverse--variance demand
shrinks by a factor of six, but the constant shift $+\tfrac12$ keeps the
position long at all variance levels; the ratio of the Stratonovich to the
It\^o position, $(0.5V_t+0.01)/0.06$, now depends on the state. Under the
Klimontovich convention ($K_1=-0.04$) the two components compete:
$\pi^\ast_t=1-0.04/V_t$ is \emph{short} in low--variance states and turns
long once $V_t>0.04$, i.e.\ once instantaneous volatility exceeds $20\%$.
Gauging how \emph{typical} such states are requires the convention's own
variance law: in the fixed--coefficient experiment \eqref{eq:theta-alpha}
gives long--run means $\theta_\alpha=\theta+\alpha\xi^2/(2\kappa)$, i.e.\
$0.045$, $0.065$, and $0.085$ for $\alpha=0,\tfrac12,1$ (long--run
volatilities of $21.2\%$, $25.5\%$, and $29.2\%$). At its own long--run mean
each position is long ($\approx1.33$, $\approx0.65$, $\approx0.53$)
but under the stationary
$\mathrm{Gamma}(2\kappa\theta_1/\xi^2,\,2\kappa/\xi^2)$ law of the CIR
process the Klimontovich short region $V_t<0.04$ carries roughly $23\%$ of
the mass. The sign of this exposure is thus determined jointly by the
noise interpretation and the volatility state, rather than by the excess
return alone, and the economically relevant range of volatility states
itself shifts with the convention.

\begin{figure}[htbp]
\centering
\includegraphics[width=0.75\linewidth]{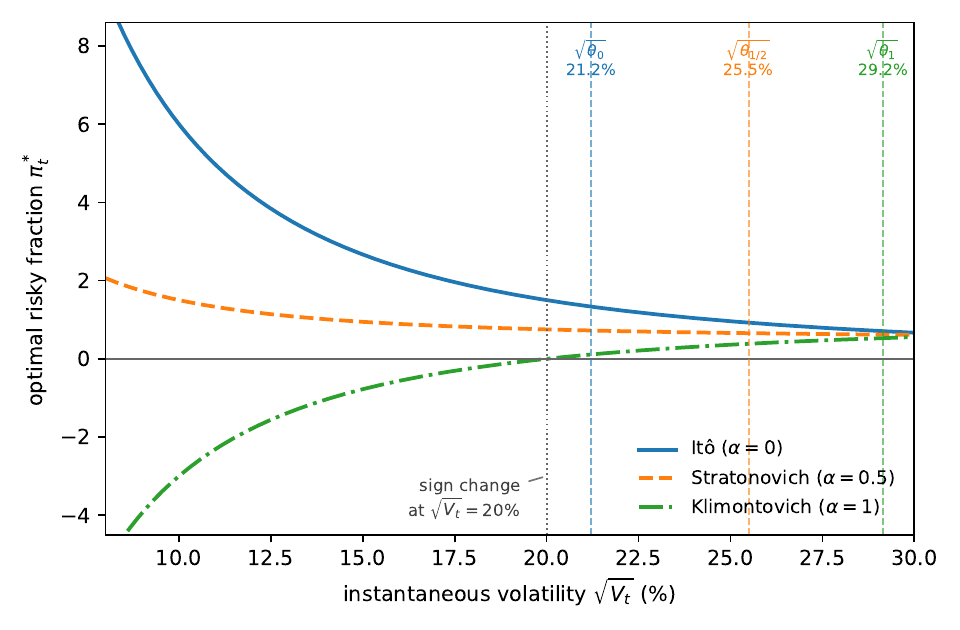}
\caption{Optimal risky fraction $\pi^*_t=\alpha+K_\alpha/V_t$ in the Heston
model as a function of instantaneous volatility $\sqrt{V_t}$, for
$\alpha\in\{0,0.5,1\}$, under $\kappa=2$, $\theta=4.5\%$, $\xi=0.4$,
$\varrho=-0.5$, $\mu-r=6\%$, holding nominal coefficients fixed across
conventions. With the empirically motivated negative correlation, the
Klimontovich position changes sign at $\sqrt{V_t}=20\%$ (dotted line): it is
short below and long above this volatility level, while the It\^o and
Stratonovich exposures remain long and diverge as volatility falls. The
colour--matched vertical dashed lines mark each convention's \emph{own}
long--run volatility $\sqrt{\theta_\alpha}$ from \eqref{eq:theta-alpha}
($21.2\%$, $25.5\%$, $29.2\%$ for $\alpha=0,0.5,1$): holding nominal
coefficients fixed shifts the stationary level of variance itself, and at
its own long--run mean each convention's position is long.}
\label{fig:heston-pi}
\end{figure}

Since $\pi_t^\ast=\alpha+K_\alpha/V_t$, the variance contribution of the
optimal policy to the wealth dynamics is
\[
(\pi_t^\ast)^2V_t
=\alpha^2V_t+2\alpha K_\alpha+\frac{K_\alpha^2}{V_t},
\]
so admissibility, the integrability requirement
$\E\int_0^T(\pi_t^\ast)^2V_t\,\dd t<\infty$ of
Definition~\ref{def:admissible}(iii) in its factor--driven form, i.e.\
hypothesis~\eqref{eq:verification-martingale} of
Theorem~\ref{thm:verification}, not a pointwise statement at individual
variance levels, involves moments of both $V_t$ (always finite for the
CIR process) and $1/V_t$. Strict positivity of the variance process and integrability of
$1/V_t$ are \emph{not} the same condition. The classical Feller condition
(\citealp{Feller1951}; see also~\citealp{cox1985theory}) for the CIR process
\eqref{eq:heston-Ito-V},
\begin{equation}\label{eq:feller-weak}
2\kappa\theta_\alpha \ge \xi^2,
\end{equation}
with $\theta_\alpha$ given by \eqref{eq:theta-alpha}, ensures that the boundary
$V=0$ is almost surely never attained, i.e.\ $\mathbb P(V_t>0\ \forall t\ge0)=1$.
This alone, however, does not guarantee that $\E[1/V_t]$ is finite. Using the
classical noncentral--chi-square transition law of the CIR process
\citep{cox1985theory}, $V_t$ is distributed as a positive constant times a
noncentral chi-square random variable with $\delta:=4\kappa\theta_\alpha/\xi^2$
degrees of freedom, whose density behaves like $z^{\delta/2-1}$ as $z\to0^+$
\citep{GoingJaeschkeYor2003}. Consequently
\begin{equation}\label{eq:cir-inverse-moment}
    \E[V_t^{-1}]<\infty
    \iff
    \delta>2
    \iff
    2\kappa\theta_\alpha > \xi^2
    \qquad(\text{strict inequality}),
\end{equation}
and at the boundary $2\kappa\theta_\alpha=\xi^2$ of~\eqref{eq:feller-weak},
$\E[1/V_t]=\infty$ for every $t>0$ even though $V_t>0$ almost surely. Under the
strict condition~\eqref{eq:cir-inverse-moment}, $t\mapsto\E[V_t^{-1}]$ is
bounded on $[0,\infty)$ (it is continuous, equals $1/V_0$ at $t=0$, and
converges to a finite stationary inverse moment as $t\to\infty$
\citealp{Dufresne2001}), so $\E\bigl[\int_0^\infty e^{-\rho t}V_t^{-1}\,\dd t\bigr]<\infty$
for every $\rho>0$. Under strict Feller, this inverse--moment bound is
sufficient for admissibility of $\pi_t^\ast$ in the sense of
Definition~\ref{def:admissible}(iii) --- and hence for
hypothesis~\eqref{eq:verification-martingale} of
Theorem~\ref{thm:verification} along $\pi_t^\ast$ --- and for the Feynman--Kac
representation~\eqref{eq:v-feynman-kac} of
Proposition~\ref{prop:factor-verification}. When $K_\alpha\neq0$, it is also
necessary because the relevant expressions contain a $K_\alpha^2/V_t$ term;
when $K_\alpha=0$, the inverse term disappears and the first moment of $V_t$
suffices. We therefore adopt the strict Feller condition
$2\kappa\theta_\alpha>\xi^2$, rather than the non-strict
inequality~\eqref{eq:feller-weak}, as a sufficient operative condition for this
example. Under this condition the full verification of
Theorem~\ref{thm:verification} for the Heston controls is completed in
Appendix~\ref{app:heston-verification}, so that
$(c_t^\ast,\pi_t^\ast)$ are optimal, not merely candidate, policies.

\section{Concluding remarks}\label{sec:conclusion}

This paper studied the classical Merton consumption--investment problem under
a family of stochastic--integral conventions, keeping preferences, the
filtration, and the self--financing rule fixed, and distinguishing carefully
between two comparative exercises. When the drift coefficient is transformed
jointly with the convention, the asset's It\^o law is unchanged and the
optimal policy is representation invariant
(Proposition~\ref{prop:invariance}): the stochastic--integral convention, by
itself, has no economic content. When instead the convention--specific drift
coefficient is held numerically fixed, changing the convention changes the
effective It\^o drift --- and hence the law of returns and the investor's
opportunity set --- and the resulting policy shifts admit closed forms under
logarithmic utility. In the constant--volatility market, the shift is
$\alpha V^{-1}d_V$, which equals $\alpha$ on every asset only when the
off--diagonal row sums of $V$ vanish --- the univariate or uncorrelated case
for $n\le3$, with further exceptions from $n=4$
(Theorem~\ref{thm:n-asset-alpha} and
Remark~\ref{rem:sensitivity-direction}); away from that restriction the
per--asset correction generally departs from $\alpha$ and can reverse sign.
Figure~\ref{fig:theta-alpha} shows
that in three independent single--asset markets the same unit shift represents
anywhere from a $48\%$ to a $153\%$ change in risky exposure. Once volatility
is driven by a correlated stochastic factor, the drift adjustment decomposes
into an own--variance component, present even without correlation, and a
covariation component scaling with the instantaneous covariation between
factor and return noise. In the Heston illustration of
Section~\ref{subsec:heston-example} the optimal risky fraction becomes
$\pi_t^\ast=\alpha+(\mu-r+\alpha\varrho\xi/2)/V_t$: under the empirically
standard negative correlation the two components compete, and the sign of the
position depends jointly on the convention and the volatility state
(Figure~\ref{fig:heston-pi}). Taken together, these results delimit precisely
when the interpretation of stochastic noise is innocuous (always), if drift
coefficients are transformed consistently, and when it is substantive:
whenever nominal coefficients that are not invariant across conventions are
treated as fixed model inputs.

These conclusions were obtained under deliberately stylized conditions, and
several of them are worth stating explicitly as limitations rather than
leaving them implicit. First, logarithmic utility was used throughout because
it eliminates intertemporal hedging demand and preserves closed--form
solutions; under general CRRA or HARA preferences the $\alpha$--correction
would interact with a nonzero hedging term, and the clean additive structure
of Theorem~\ref{thm:n-asset-alpha} need not survive. Second, the
state--dependent analysis of Section~\ref{sec:factor-correlated-noise} was
carried out for a single stochastic factor driving both volatility and, in the
Heston case, its own long--run mean; multi--factor and rough--volatility
specifications, in which volatility is driven by processes with different
memory properties, are not covered by the present analysis
\citep{gatheral2018volatility,bayer2016pricing}. Third, the market was assumed
frictionless and continuously tradable, so transaction costs and portfolio
constraints \citep{DavisNorman1990,ShreveSoner1994} were left aside entirely.
Finally, the paper is deliberately agnostic about which value of $\alpha$ best
describes any particular market: it isolates the comparative--statics effect
of the interpretation itself rather than proposing a procedure for estimating
it.

This last point is also the most immediate direction for future work, and it
faces a genuine identification problem. With constant volatility, return data
identify only the effective It\^o drift $\mu+\alpha\sigma^2$, not $\mu$ and
$\alpha$ separately; in the factor--driven model, freely parametrized drift
functions can absorb the $\alpha$--corrections entirely. Any empirical
strategy for $\alpha$ therefore requires structural restrictions, for
example, a microstructural mechanism linking latency or local smoothing to a
specific convention, observations across sampling frequencies combined with an
aggregation theory, or moment restrictions (such as the reversibility and
coarse--graining properties studied in \citet{ayala2025reversibility}) that
depend on $\alpha$ but cannot be reproduced by adjusting the drift.
Pre--averaging and related high--frequency procedures
\citep{HansenLunde2006,AitSahaliaYu2009,JacodEtAl2009} motivate asking whether
particular sampling mechanisms converge to effective equations associated with
different conventions, but establishing such a mapping is an open problem
rather than an available identification result. Beyond this, extending the analysis to general
HARA utility, to multi--factor or rough stochastic volatility, and to markets
with trading frictions are all natural next steps that would test how much of
the clean structure identified here survives outside the log--utility,
single--factor setting.

\backmatter

\section*{Acknowledgements}
M.~Ayala is grateful to Johannes Zimmer for stimulating discussions that
motivated the inclusion of the Heston stochastic--volatility example.

\section*{Declarations}

\begin{itemize}
\item \textbf{Funding.} The authors received no specific grant from any funding agency, commercial or not-for-profit sector, for this research.
\item \textbf{Conflict of interest.} The authors declare that they have no conflict of interest.
\item \textbf{Data availability.} No datasets were generated or analysed during the current study; all numerical illustrations are fully reproducible from the parameter values reported in the text.
\item \textbf{Code availability.} The script generating the numerical illustrations is available from the authors upon request.
\item \textbf{Ethics approval.} Not applicable.
\item \textbf{Consent to participate.} Not applicable.
\item \textbf{Consent for publication.} Not applicable.
\item \textbf{Author contributions.} Both authors contributed equally to the conception and design of the study, the analysis, and the writing of the manuscript, and both approved the final version.
\end{itemize}

\begin{appendices}

\section{A dictionary between noise interpretations}\label{app:dictionary}

In the main text we formulate asset--price dynamics using stochastic differential
equations written in the $\alpha$--interpretation, where the stochastic integral
$\circ_\alpha$ interpolates between It\^o ($\alpha=0$), Stratonovich
($\alpha=\tfrac12$), and Klimontovich ($\alpha=1$) conventions.  For the
consumption--investment problem studied in this paper, we focus on geometric
Brownian motion, which corresponds to the simplest instance of multiplicative
noise and allows for closed--form solutions.

The purpose of this appendix is to place this choice within a broader modelling
framework.  We collect general conversion formulas that allow one to translate
an SDE written under a given interpretation into an equivalent SDE under another
interpretation, for diffusion coefficients beyond the geometric Brownian case.
These dictionaries enable the reader to understand how changes in noise
interpretation modify effective drift terms for more general stock price
dynamics, and thus how the analysis of the main text can be extended to
alternative multiplicative or state--dependent volatility structures.  The drift
corrections derived here underlie, as a special case, the shifts used in
Sections~\ref{subsec:merton-classical} and~\ref{sec:n-asset-alpha}.

\subsection{\texorpdfstring{General $\alpha\to\gamma$ conversion}{General alpha-to-gamma conversion}}

Let $W=(W_t)_{t\ge0}$ be an $m$–dimensional Brownian motion on the filtered
probability space specified in Section~\ref{sec:math-prelim}, and consider the
$d$–dimensional SDE
\begin{equation}\label{eq:SDE-alpha-general}
\dd X_t
  = b(X_t)\,\dd t + \Sigma(X_t)\,\circ_\alpha\dd W_t,
\end{equation}
where $b:\R^d\to\R^d$ and $\Sigma:\R^d\to\R^{d\times m}$ are sufficiently smooth
(e.g.\ $C^2$ with bounded derivatives), and $\circ_\alpha$ denotes integration in
the $\alpha$–interpretation.

Writing \eqref{eq:SDE-alpha-general} in It\^o form corresponds to the special
case $\gamma=0$ in the following general dictionary.

\begin{proposition}[Conversion between $\alpha$– and $\gamma$–interpretations]
\label{prop:alpha-gamma-dictionary}
Fix $\alpha,\gamma\in[0,1]$ and let $X$ solve \eqref{eq:SDE-alpha-general}.
Then $X$ also solves
\begin{equation}\label{eq:SDE-gamma-general}
\dd X_t
  = \widetilde b(X_t)\,\dd t + \Sigma(X_t)\,\circ_\gamma\dd W_t,
\end{equation}
where the drift $\widetilde b$ is given componentwise by
\begin{equation}\label{eq:drift-shift-general}
\widetilde b_i(x)
  = b_i(x)
    + (\alpha-\gamma)\sum_{k=1}^m\sum_{j=1}^d
        \Sigma_{jk}(x)\,\partial_{x_j}\Sigma_{ik}(x),
\qquad i=1,\dots,d.
\end{equation}
Equivalently, the diffusion coefficient $\Sigma(x)$ is unchanged, while the
drift is shifted by
\[
\widetilde b(x) = b(x) + (\alpha-\gamma)\,c(x),\qquad
c_i(x):=\sum_{k,j}\Sigma_{jk}(x)\,\partial_{x_j}\Sigma_{ik}(x).
\]
\end{proposition}

\begin{proof}[Proof sketch]
The proof is standard and relies on writing both $\alpha$– and $\gamma$–integrals
as limits of Riemann sums with evaluation at intermediate points of each time
interval.  For smooth $\Sigma$, the difference between the two evaluations can be
expanded to first order and expressed in terms of the quadratic covariation of
$X$ and $W$.  This yields the correction term proportional to
$(\alpha-\gamma)\sum_{j,k}\Sigma_{jk}\partial_{x_j}\Sigma_{ik}$.
Rigorous derivations can be found, for example, in classical references on
generalised stochastic integrals.
\end{proof}

The It\^o form of \eqref{eq:SDE-alpha-general} corresponds to $\gamma=0$:
\begin{equation}\label{eq:alpha-to-ito-appendix}
\dd X_t
  = \Bigl(b(X_t) + \alpha\,c(X_t)\Bigr)\dd t + \Sigma(X_t)\,\dd W_t,
\qquad
c_i(x)=\sum_{k,j}\Sigma_{jk}(x)\,\partial_{x_j}\Sigma_{ik}(x).
\end{equation}
Similarly, the Stratonovich form is obtained by taking $\gamma=\tfrac12$, and
the Klimontovich form by taking $\gamma=1$.

\begin{remark}[Correlated Brownian motions]\label{rem:correlated-appendix}
The conversion formulas above are stated for an $m$--dimensional Brownian motion
$W$ with identity covariance, i.e.\ $[W]_t = I_m\,t$. This entails no loss of
generality.

Indeed, if $\widetilde W$ is an $m$--dimensional Brownian motion with constant
covariance matrix $R$, one may write $\widetilde W = C W$ for some deterministic
matrix $C$ satisfying $C C^\top = R$. Rewriting
\eqref{eq:SDE-alpha-general} in terms of $W$ replaces the diffusion coefficient
$\Sigma(x)$ by $\Sigma(x)C$, and the conversion formula
\eqref{eq:drift-shift-general} applies verbatim.

In particular, in the factor--driven model of
Section~\ref{sec:factor-correlated-noise}, the additional drift term
$\alpha\,\varrho\,\sigma'(x)\nu(x)$ arises precisely from this covariance
structure.
\end{remark}

\subsection{Diagonal--multiplicative noise}

In the main body of the paper we work with a class of multiplicative diffusion
models in which each risky asset is affected by a common vector of Brownian
shocks, but with volatility proportional to its own price level.
Let $S_t=(S_t^1,\dots,S_t^n)\in(0,\infty)^n$ denote the vector of risky asset
prices, let $\mu\in\R^n$ be a constant drift vector, and let
$\Gamma\in\R^{n\times n}$ be a constant volatility loading matrix.
We consider the $\alpha$--interpreted SDE
\begin{equation}\label{eq:diag-mult-alpha}
  \dd S_t
  =
  \mathrm{D}(S_t)\Bigl(\mu\,\dd t + \Gamma\,\circ_\alpha\dd W_t\Bigr),
\end{equation}
where $\mathrm{D}(S)=\diag(S^1,\dots,S^n)$ and $W_t$ is an $n$--dimensional
Brownian motion.

Equivalently, in components,
\[
  \dd S_t^i
  =
  S_t^i
  \left(
    \mu_i\,\dd t + \sum_{k=1}^n \Gamma_{ik}\,\circ_\alpha\dd W_t^k
  \right),
  \qquad i=1,\dots,n.
\]

The diffusion coefficient in \eqref{eq:diag-mult-alpha} is the matrix--valued
function
\[
  B(S) := \mathrm{D}(S)\Gamma \in \R^{n\times n},
\]
with components
\[
  B_{ik}(S) = S^i\,\Gamma_{ik}.
\]
In particular, the $i$--th \emph{row} of $B(S)$ depends only on the component
$S^i$, a structural property that will be crucial in the computation of the
$\alpha$--to--It\^o drift correction.

Define the constant covariance matrix
\[
  V := \Gamma\Gamma^\top \in \R^{n\times n},
  \qquad
  V_{ij} = \sum_{k=1}^n \Gamma_{ik}\Gamma_{jk}.
\]
Applying the general conversion formula
\eqref{eq:alpha-to-ito-appendix} to \eqref{eq:diag-mult-alpha}, we obtain that
the It\^o drift correction has components
\[
  c_i(S)
  :=
  \sum_{j=1}^n\sum_{k=1}^n
  B_{jk}(S)\,\partial_{S^j} B_{ik}(S).
\]
Since $B_{ik}(S)=S^i\Gamma_{ik}$, we have
\[
  \partial_{S^j}B_{ik}(S)=\Gamma_{ik}\,\delta_{ij},
\]
and therefore
\[
  c_i(S)
  =
  \sum_{k=1}^n B_{ik}(S)\Gamma_{ik}
  =
  S^i \sum_{k=1}^n \Gamma_{ik}^2
  =
  S^i\,V_{ii},
  \qquad i=1,\dots,n.
\]

The $\alpha$--interpreted SDE \eqref{eq:diag-mult-alpha} is thus equivalent to
the It\^o SDE
\begin{equation}\label{eq:diag-mult-ito}
  \dd S_t
  =
  \mathrm{D}(S_t)\Bigl(\mu^{\mathrm{Ito}}\,\dd t + \Gamma\,\dd W_t\Bigr),
  \qquad
  \mu^{\mathrm{Ito}} := \mu + \alpha\,\diag(V).
\end{equation}
This coincides exactly with Proposition~\ref{prop:alpha-to-ito} in
Section~\ref{sec:n-asset-alpha}.

More generally, changing the interpretation parameter from $\alpha$ to
$\gamma$ leaves the diffusion term unchanged and shifts the drift according to
\[
  \mu^{\mathrm{Ito}}
  \;\longmapsto\;
  \mu^{\mathrm{Ito}} + (\alpha-\gamma)\,\diag(V).
\]
In particular, moving from It\^o to Stratonovich adds
$\tfrac12\diag(V)$ to the drift, while moving from It\^o to Klimontovich adds
$\diag(V)$.

\begin{remark}
From the perspective of the consumption--investment problem, this dictionary
justifies treating different noise interpretations (\emph{with
coefficients held fixed}) as changes in effective expected returns, while
keeping both the volatility structure and the self--financing constraint
unchanged. This mechanism is responsible for the shift
\[
  \theta^*_\alpha
  =
  V^{-1}(\mu-r\mathbf{1}) + \alpha\,V^{-1}d_V
\]
in the optimal risky portfolio derived in Theorem~\ref{thm:n-asset-alpha}.
Conversely, the same dictionary shows that transforming the drift by
$-(\alpha-\gamma)c(x)$ when moving from convention $\alpha$ to convention
$\gamma$ leaves the solution process, and hence the entire control problem,
unchanged (Proposition~\ref{prop:invariance}).
\end{remark}

\section{Optimality of the Heston controls under the strict Feller condition}
\label{app:heston-verification}

This appendix verifies, for the Heston specialization of
Section~\ref{subsec:heston-example}, all hypotheses of the verification
theorem (Theorem~\ref{thm:verification}) that were left open in the general
factor--driven Proposition~\ref{prop:factor-verification}. The only
structural input is the \emph{strict} Feller condition
\begin{equation}\label{eq:app-strict-feller}
2\kappa\theta_\alpha>\xi^2
\qquad\Longleftrightarrow\qquad
\delta:=\frac{4\kappa\theta_\alpha}{\xi^2}>2,
\end{equation}
already adopted in~\eqref{eq:cir-inverse-moment}; no new assumption is
introduced. The variance factor is the CIR process~\eqref{eq:heston-Ito-V},
written in It\^o form with converted drift
$b_\alpha(v)=\kappa(\theta_\alpha-v)$ and diffusion $\nu(v)=\xi\sqrt v$, on
the state space $(0,\infty)$, with generator
$\mathcal L=\tfrac12\xi^2 v\,\partial_v^2+\kappa(\theta_\alpha-v)\,\partial_v$.
From $\mu_{\mathrm{eff}}(v)-r=K_\alpha+\alpha v$
(see~\eqref{eq:heston-controls}), the source term of the resolvent
equation~\eqref{eq:v-ode} is
\begin{equation}\label{eq:app-g}
g(v)=\ln\rho-1+\frac r\rho
   +\frac{1}{2\rho}\Bigl(\alpha^2 v+2\alpha K_\alpha+\frac{K_\alpha^2}{v}\Bigr),
\end{equation}
which is real--analytic on $(0,\infty)$ and satisfies
$|g(v)|\le C_g\,(1+v+v^{-1})$ for a constant $C_g$ depending only on the
parameters. Two moment facts, both consequences
of~\eqref{eq:app-strict-feller}, are used repeatedly: by mean reversion,
$\E_x[V_s]=\theta_\alpha+(x-\theta_\alpha)e^{-\kappa s}$, so
$m_1(x):=\sup_{s\ge0}\E_x[V_s]<\infty$ for each $x\in(0,\infty)$; and, by
the boundedness statement following~\eqref{eq:cir-inverse-moment} (see also
\citealp{Dufresne2001}),
\begin{equation}\label{eq:app-invmom}
m_{-1}:=\sup_{s\ge0}\E[V_s^{-1}]<\infty.
\end{equation}

\begin{proposition}\label{prop:app-heston-optimality}
Assume the strict Feller condition~\eqref{eq:app-strict-feller}. Then the
Feynman--Kac functional~\eqref{eq:v-feynman-kac} is a classical $C^2$
solution of the resolvent equation~\eqref{eq:v-ode} on $(0,\infty)$, and all
hypotheses of Theorem~\ref{thm:verification} hold for the Heston model.
Consequently the controls~\eqref{eq:heston-controls},
\[
c_t^\ast=\rho A_t,
\qquad
\pi_t^\ast=\alpha+\frac{K_\alpha}{V_t},
\]
are optimal --- not merely candidate optima --- for every $\alpha\in[0,1]$
satisfying~\eqref{eq:app-strict-feller}.
\end{proposition}

\begin{proof}
\emph{Step 1: the Feynman--Kac functional is a classical solution.}
Under~\eqref{eq:app-strict-feller} the endpoint $0$ is an entrance boundary
and $+\infty$ a natural boundary, so the process stays in $(0,\infty)$ for
all $t\ge0$ almost surely \citep{Feller1951,cox1985theory}, and
$\mathcal L$ is the generator of a regular one--dimensional diffusion with
inaccessible boundaries, to which the classical Green--function theory of
the $\rho$--resolvent applies \citep[\S4.11]{ItoMcKean1974};
\citep[\S II.10--11]{BorodinSalminen2002}. Convergence of the two Green
integrals is equivalent to the probabilistic integrability
defining~\eqref{eq:v-feynman-kac}, which holds because
$|g(v)|\le C_g(1+v+v^{-1})$ gives
\begin{equation}\label{eq:app-fk-bound}
\E_x\!\int_0^\infty e^{-\rho s}|g(V_s)|\,\dd s
\le \frac{C_g}{\rho}\bigl(1+m_1(x)+m_{-1}\bigr)<\infty ;
\end{equation}
strict Feller is thus sufficient for the integrability hypothesis of
Proposition~\ref{prop:factor-verification} (its necessity when
$K_\alpha\neq0$ is discussed in Section~\ref{subsec:heston-example}). Hence
the resolvent $v:=R_\rho g$ is well defined and coincides
with~\eqref{eq:v-feynman-kac}. Since $g$ and the speed and scale densities
are real--analytic on $(0,\infty)$, and the homogeneous solutions
$\varphi_\rho$ (increasing) and $\psi_\rho$ (decreasing) of
$\mathcal Lu-\rho u=0$ are confluent hypergeometric (Kummer) functions,
real--analytic on $(0,\infty)$ \citep[Appendix~1]{BorodinSalminen2002},
differentiating the Green--function representation of $R_\rho g$ twice 
(once with respect to the scale function and once with respect to the speed
measure) shows that $v\in C^2\bigl((0,\infty)\bigr)$ and
solves~\eqref{eq:v-ode} classically \citep[\S4.11]{ItoMcKean1974}. The
hypothesis $v\in C^2$ of Theorem~\ref{thm:verification} is therefore
established, not assumed. Only this particular solution enters the argument
below; uniqueness for the inhomogeneous ODE~\eqref{eq:v-ode} is neither
claimed nor required.

\emph{Step 2: boundary asymptotics of $v$.} The leading behaviour of $v$ at
the two boundaries follows from $(\mathcal L-\rho)v=-g$ by matching the
dominant terms of $g$ in~\eqref{eq:app-g}. Near $0$, when $K_\alpha\neq0$
the dominant term is $K_\alpha^2/(2\rho x)$; the trial $v(x)=A\log x$ gives
$(\mathcal L-\rho)(A\log x)
=A\bigl(\kappa\theta_\alpha-\tfrac12\xi^2\bigr)/x-\kappa A-\rho A\log x$,
and matching the coefficients of $1/x$ yields
\[
A=-\frac{K_\alpha^2}{2\rho\bigl(\kappa\theta_\alpha-\tfrac12\xi^2\bigr)}<0,
\]
the denominator being positive exactly under strict
Feller~\eqref{eq:app-strict-feller}. Near $+\infty$, when $\alpha\neq0$ the
dominant term is $\alpha^2x/(2\rho)$; the trial $v(x)=Bx$ gives
$(\mathcal L-\rho)(Bx)=B\kappa\theta_\alpha-B(\kappa+\rho)x$, whence
$B=\alpha^2/\bigl(2\rho(\rho+\kappa)\bigr)$. The resolvent solution inherits
exactly these particular asymptotics because the Green--function
construction excludes both homogeneous modes: adding $\varphi_\rho$ would
introduce exponential growth at the natural boundary $+\infty$, and adding
$\psi_\rho$ a power divergence $\psi_\rho(x)\sim C_\psi x^{1-\delta/2}$,
$1-\delta/2<0$, at the entrance boundary $0$, incompatible with the
behaviour just derived. Therefore
\begin{equation}\label{eq:app-v-growth}
v(x)=A\log x+O(1),\quad v'(x)\sim\frac Ax\quad(x\downarrow0);
\qquad
v'(x)\to\frac{\alpha^2}{2\rho(\rho+\kappa)}\quad(x\to\infty),
\end{equation}
so that $v(x)=O(1+|\log x|+x)$ and, since $v'$ is continuous on compact
subsets of $(0,\infty)$,
$\nu(x)^2v'(x)^2=\xi^2x\,v'(x)^2\le C\,(1+x+x^{-1})$ on $(0,\infty)$. In the
degenerate cases these singular features disappear --- $v$ is bounded near
$0$ when $K_\alpha=0$, and has no linear growth at $\infty$ when $\alpha=0$
--- and the bounds above hold a fortiori.

\emph{Step 3: hypothesis (ii).} By Step~2 and the two moment facts,
\[
\E\!\int_0^T\!\nu(V_t)^2 v'(V_t)^2\,\dd t
\le CT\bigl(1+m_1(V_0)+m_{-1}\bigr)<\infty
\qquad\text{for every }T>0,
\]
which is the factor--integrability
condition~\eqref{eq:verification-v-integrability}. For the terminal
condition, the Markov property and the bound~\eqref{eq:app-fk-bound} give
\begin{equation}\label{eq:app-Ev}
\E\bigl[|v(V_T)|\bigr]
\le \E\!\int_0^\infty e^{-\rho u}|g(V_{T+u})|\,\dd u
\le \frac{C_g}{\rho}\bigl(1+m_1(V_0)+m_{-1}\bigr)<\infty
\end{equation}
uniformly in $T\ge0$; in particular $\E[v(V_T)^-]<\infty$ for every $T$ (which is condition~\eqref{eq:verification-v-terminal}) and
$e^{-\rho T}\E[v(V_T)]\to0$.

\emph{Step 4: admissibility of the candidate and transversality.} Under the
candidate
$(c^\ast_t,\pi^\ast_t)=(\rho A_t,\,(\mu_{\mathrm{eff}}(V_t)-r)/V_t)$, the
wealth SDE~\eqref{eq:wealth-factor-4} and It\^o's formula give
\begin{equation}\label{eq:app-lna}
\dd\ln A_t^\ast
=\Bigl(r-\rho+\tfrac12\bigl(\alpha^2V_t+2\alpha K_\alpha
  +K_\alpha^2V_t^{-1}\bigr)\Bigr)\dd t
 +\frac{\mu_{\mathrm{eff}}(V_t)-r}{\sqrt{V_t}}\,\dd W_t^{S},
\end{equation}
so wealth stays strictly positive. The martingale
condition~\eqref{eq:verification-martingale},
$\E\int_0^T(\pi^\ast_t)^2V_t\,\dd t<\infty$, was verified
in~\eqref{eq:cir-inverse-moment}. The drift in~\eqref{eq:app-lna} is bounded
in expectation uniformly in $t$ by the two moment facts, and the martingale
part has $L^2$ norm of order $\sqrt T$ (its expected quadratic variation is
$O(T)$); hence $\E|\ln A_T^\ast|=O(T)$. This yields the terminal
condition~\eqref{eq:adm-terminal},
$e^{-\rho T}\E[(\ln A_T^\ast)^-]\to0$, and, since
$\log c_t^\ast=\log\rho+\ln A_t^\ast$, also the positive--part utility
condition~\eqref{eq:adm-utility-pos}:
$\E[(\log c_t^\ast)^+]=O(1+t)$, so
$\int_0^\infty e^{-\rho t}\E[(\log c_t^\ast)^+]\,\dd t<\infty$. The
candidate is therefore admissible. Moreover $\E[\ln A_T^\ast]=O(T)$
and~\eqref{eq:app-Ev} give
\[
\lim_{T\to\infty}\E\bigl[J(A_T^\ast,V_T,T)\bigr]
=\lim_{T\to\infty}e^{-\rho T}
 \bigl(\rho^{-1}\E[\ln A_T^\ast]+\E[v(V_T)]\bigr)=0,
\]
the transversality condition~\eqref{eq:verification-transversality}.

\emph{Step 5: the liminf condition.} Let $(c_t,\theta_t)$ be any admissible
control with wealth $a_t$. The variance factor is \emph{autonomous}, its
law does not depend on the control, so~\eqref{eq:app-Ev} applies verbatim
and $e^{-\rho T}\E[v(V_T)]\to0$. For the wealth part,
$e^{-\rho T}\rho^{-1}\E[\ln a_T]
\ge-\rho^{-1}e^{-\rho T}\E[(\ln a_T)^-]\to0$
by the terminal condition~\eqref{eq:adm-terminal} built into
Definition~\ref{def:admissible}(iv). Adding the two limits,
$\liminf_{T\to\infty}\E[J(a_T,V_T,T)]\ge0$, which is
condition~\eqref{eq:verification-liminf} for every admissible control.

Steps 1--5 verify every hypothesis of Theorem~\ref{thm:verification}
under~\eqref{eq:app-strict-feller}, and the conclusion follows.
\end{proof}

\begin{remark}
The moment inputs used above, $\sup_s\E[V_s]<\infty$,
$\sup_s\E[V_s^{-1}]<\infty$, and the resulting uniform bound on
$\E[v(V_T)]$, all reduce to the single strict Feller inequality
$\delta>2$. The function $v$ itself is in general \emph{not} bounded: for
$K_\alpha\neq0$ it carries a logarithmic singularity at the entrance
boundary $0$, and for $\alpha\neq0$ it grows linearly at $\infty$ (Step~2).
What closes the argument is not boundedness of $v$ but boundedness of the
CIR moments of $V_T$, which is why the strict, rather than the weak, 
Feller condition is the operative requirement.
\end{remark}

\end{appendices}

\begingroup
\interlinepenalty=10000
\bibliography{biblio}%


\begin{thebibliography}{36}
\ifx \bisbn   \undefined \def \bisbn  #1{ISBN #1}\fi
\ifx \binits  \undefined \def \binits#1{#1}\fi
\ifx \bauthor  \undefined \def \bauthor#1{#1}\fi
\ifx \batitle  \undefined \def \batitle#1{#1}\fi
\ifx \bjtitle  \undefined \def \bjtitle#1{#1}\fi
\ifx \bvolume  \undefined \def \bvolume#1{\textbf{#1}}\fi
\ifx \byear  \undefined \def \byear#1{#1}\fi
\ifx \bissue  \undefined \def \bissue#1{#1}\fi
\ifx \bfpage  \undefined \def \bfpage#1{#1}\fi
\ifx \blpage  \undefined \def \blpage #1{#1}\fi
\ifx \burl  \undefined \def \burl#1{\textsf{#1}}\fi
\ifx \doiurl  \undefined \def \doiurl#1{\url{https://doi.org/#1}}\fi
\ifx \betal  \undefined \def \betal{\textit{et al.}}\fi
\ifx \binstitute  \undefined \def \binstitute#1{#1}\fi
\ifx \binstitutionaled  \undefined \def \binstitutionaled#1{#1}\fi
\ifx \bctitle  \undefined \def \bctitle#1{#1}\fi
\ifx \beditor  \undefined \def \beditor#1{#1}\fi
\ifx \bpublisher  \undefined \def \bpublisher#1{#1}\fi
\ifx \bbtitle  \undefined \def \bbtitle#1{#1}\fi
\ifx \bedition  \undefined \def \bedition#1{#1}\fi
\ifx \bseriesno  \undefined \def \bseriesno#1{#1}\fi
\ifx \blocation  \undefined \def \blocation#1{#1}\fi
\ifx \bsertitle  \undefined \def \bsertitle#1{#1}\fi
\ifx \bsnm \undefined \def \bsnm#1{#1}\fi
\ifx \bsuffix \undefined \def \bsuffix#1{#1}\fi
\ifx \bparticle \undefined \def \bparticle#1{#1}\fi
\ifx \barticle \undefined \def \barticle#1{#1}\fi
\bibcommenthead
\ifx \bconfdate \undefined \def \bconfdate #1{#1}\fi
\ifx \botherref \undefined \def \botherref #1{#1}\fi
\ifx \url \undefined \def \url#1{\textsf{#1}}\fi
\ifx \bchapter \undefined \def \bchapter#1{#1}\fi
\ifx \bbook \undefined \def \bbook#1{#1}\fi
\ifx \bcomment \undefined \def \bcomment#1{#1}\fi
\ifx \oauthor \undefined \def \oauthor#1{#1}\fi
\ifx \citeauthoryear \undefined \def \citeauthoryear#1{#1}\fi
\ifx \endbibitem  \undefined \def \endbibitem {}\fi
\ifx \bconflocation  \undefined \def \bconflocation#1{#1}\fi
\ifx \arxivurl  \undefined \def \arxivurl#1{\textsf{#1}}\fi
\csname PreBibitemsHook\endcsname

\bibitem[\protect\citeauthoryear{Ayala et~al.}{2025}]{ayala2025reversibility}
\begin{botherref}
\oauthor{\bsnm{Ayala}, \binits{M.}},
\oauthor{\bsnm{Dirr}, \binits{N.}},
\oauthor{\bsnm{Pavliotis}, \binits{G.A.}},
\oauthor{\bsnm{Zimmer}, \binits{J.}}:
Reversibility, covariance and coarse-graining for langevin dynamics: On the choice of multiplicative noise.
arXiv preprint arXiv:2511.03347
(2025)
\end{botherref}
\endbibitem

\bibitem[\protect\citeauthoryear{Aït-Sahalia and Yu}{2009}]{AitSahaliaYu2009}
\begin{barticle}
\bauthor{\bsnm{Aït-Sahalia}, \binits{Y.}},
\bauthor{\bsnm{Yu}, \binits{J.}}:
\batitle{High frequency market microstructure noise estimates and liquidity measures}.
\bjtitle{The Annals of Applied Statistics}
\bvolume{3}(\bissue{1}),
\bfpage{422}--\blpage{457}
(\byear{2009})
\doiurl{10.1214/08-AOAS203}
\end{barticle}
\endbibitem

\bibitem[\protect\citeauthoryear{Bender}{2003}]{Bender2003}
\begin{botherref}
\oauthor{\bsnm{Bender}, \binits{C.}}:
An it{\^o} formula for a fractional stratonovich type integral with arbitrary hurst parameter and stratonovich self-financing arbitrage.
Preprint, Department of Mathematics, University of Konstanz
(2003).
Available at \url{https://www.math.uni-konstanz.de/~kohlmann/ftp/dp02_07.pdf}
\end{botherref}
\endbibitem

\bibitem[\protect\citeauthoryear{Bayer et~al.}{2016}]{bayer2016pricing}
\begin{barticle}
\bauthor{\bsnm{Bayer}, \binits{C.}},
\bauthor{\bsnm{Friz}, \binits{P.}},
\bauthor{\bsnm{Gatheral}, \binits{J.}}:
\batitle{Pricing under rough volatility}.
\bjtitle{Quantitative Finance}
\bvolume{16}(\bissue{6}),
\bfpage{887}--\blpage{904}
(\byear{2016})
\end{barticle}
\endbibitem

\bibitem[\protect\citeauthoryear{Borodin and Salminen}{2002}]{BorodinSalminen2002}
\begin{bbook}
\bauthor{\bsnm{Borodin}, \binits{A.N.}},
\bauthor{\bsnm{Salminen}, \binits{P.}}:
\bbtitle{Handbook of Brownian Motion --- Facts and Formulae},
\bedition{2nd} edn.
\bsertitle{Probability and Its Applications}.
\bpublisher{Birkh\"auser},
\blocation{Basel}
(\byear{2002})
\end{bbook}
\endbibitem

\bibitem[\protect\citeauthoryear{Cheridito}{2003}]{cheridito2003arbitrage}
\begin{barticle}
\bauthor{\bsnm{Cheridito}, \binits{P.}}:
\batitle{Arbitrage in fractional brownian motion models}.
\bjtitle{Finance and Stochastics}
\bvolume{7}(\bissue{4}),
\bfpage{533}--\blpage{553}
(\byear{2003})
\doiurl{10.1007/s007800300101}
\end{barticle}
\endbibitem

\bibitem[\protect\citeauthoryear{Cox et~al.}{1985}]{cox1985theory}
\begin{barticle}
\bauthor{\bsnm{Cox}, \binits{J.C.}},
\bauthor{\bsnm{Ingersoll}, \binits{J.E.}},
\bauthor{\bsnm{Ross}, \binits{S.A.}}:
\batitle{A theory of the term structure of interest rates}.
\bjtitle{Econometrica}
\bvolume{53}(\bissue{2}),
\bfpage{385}--\blpage{407}
(\byear{1985})
\end{barticle}
\endbibitem

\bibitem[\protect\citeauthoryear{Cartea et~al.}{2015}]{CarteaJaimungalPenalva2015}
\begin{bbook}
\bauthor{\bsnm{Cartea}, \binits{{\'A}.}},
\bauthor{\bsnm{Jaimungal}, \binits{S.}},
\bauthor{\bsnm{Penalva}, \binits{J.}}:
\bbtitle{Algorithmic and High-Frequency Trading}.
\bpublisher{Cambridge University Press},
\blocation{Cambridge}
(\byear{2015})
\end{bbook}
\endbibitem

\bibitem[\protect\citeauthoryear{Constantinides}{1990}]{Constantinides1990}
\begin{barticle}
\bauthor{\bsnm{Constantinides}, \binits{G.M.}}:
\batitle{Habit formation: A resolution of the equity premium puzzle}.
\bjtitle{Journal of Political Economy}
\bvolume{98}(\bissue{3}),
\bfpage{519}--\blpage{543}
(\byear{1990})
\doiurl{10.1086/261693}
\end{barticle}
\endbibitem

\bibitem[\protect\citeauthoryear{Chacko and Viceira}{2005}]{ChackoViceira2005}
\begin{barticle}
\bauthor{\bsnm{Chacko}, \binits{G.}},
\bauthor{\bsnm{Viceira}, \binits{L.M.}}:
\batitle{Dynamic consumption and portfolio choice with stochastic volatility in incomplete markets}.
\bjtitle{The Review of Financial Studies}
\bvolume{18}(\bissue{4}),
\bfpage{1369}--\blpage{1402}
(\byear{2005})
\end{barticle}
\endbibitem

\bibitem[\protect\citeauthoryear{Davis and Norman}{1990}]{DavisNorman1990}
\begin{barticle}
\bauthor{\bsnm{Davis}, \binits{M.H.A.}},
\bauthor{\bsnm{Norman}, \binits{A.R.}}:
\batitle{Portfolio selection with transaction costs}.
\bjtitle{Mathematics of Operations Research}
\bvolume{15}(\bissue{4}),
\bfpage{676}--\blpage{713}
(\byear{1990})
\end{barticle}
\endbibitem

\bibitem[\protect\citeauthoryear{dos Reis and Platonov}{2021}]{DosReisPlatonov2021}
\begin{barticle}
\bauthor{\bsnm{Reis}, \binits{G.}},
\bauthor{\bsnm{Platonov}, \binits{V.}}:
\batitle{On the relation between stratonovich and it{\^o} integrals with functional integrands of conditional measure flows}.
\bjtitle{arXiv preprint arXiv:2111.03523}
(\byear{2021})
\doiurl{10.48550/arXiv.2111.03523}
\end{barticle}
\endbibitem

\bibitem[\protect\citeauthoryear{Dufresne}{2001}]{Dufresne2001}
\begin{botherref}
\oauthor{\bsnm{Dufresne}, \binits{D.}}:
The integrated square-root process.
Technical report,
University of Montreal
(2001).
Working paper
\end{botherref}
\endbibitem

\bibitem[\protect\citeauthoryear{Feller}{1951}]{Feller1951}
\begin{barticle}
\bauthor{\bsnm{Feller}, \binits{W.}}:
\batitle{Two singular diffusion problems}.
\bjtitle{Annals of Mathematics}
\bvolume{54}(\bissue{1}),
\bfpage{173}--\blpage{182}
(\byear{1951})
\end{barticle}
\endbibitem

\bibitem[\protect\citeauthoryear{Fabozzi et~al.}{2011}]{FabozziFocardiJonas2011}
\begin{barticle}
\bauthor{\bsnm{Fabozzi}, \binits{F.J.}},
\bauthor{\bsnm{Focardi}, \binits{S.M.}},
\bauthor{\bsnm{Jonas}, \binits{C.}}:
\batitle{High-frequency trading: Methodologies and market impact}.
\bjtitle{Review of Futures Markets}
\bvolume{19}(\bissue{Special Issue}),
\bfpage{7}--\blpage{38}
(\byear{2011}).
\bcomment{Available at \url{https://www.sergiofocardi.net/Papers/HighFrequencyTrading_Methodologies.pdf}}
\end{barticle}
\endbibitem

\bibitem[\protect\citeauthoryear{Fleming and Soner}{2006}]{FlemingSoner2006}
\begin{bbook}
\bauthor{\bsnm{Fleming}, \binits{W.H.}},
\bauthor{\bsnm{Soner}, \binits{H.M.}}:
\bbtitle{Controlled Markov Processes and Viscosity Solutions},
\bedition{2nd} edn.
\bsertitle{Stochastic Modelling and Applied Probability}.
\bpublisher{Springer},
\blocation{New York}
(\byear{2006})
\end{bbook}
\endbibitem

\bibitem[\protect\citeauthoryear{Gatheral et~al.}{2018}]{gatheral2018volatility}
\begin{barticle}
\bauthor{\bsnm{Gatheral}, \binits{J.}},
\bauthor{\bsnm{Jaisson}, \binits{T.}},
\bauthor{\bsnm{Rosenbaum}, \binits{M.}}:
\batitle{Volatility is rough}.
\bjtitle{Quantitative Finance}
\bvolume{18}(\bissue{6}),
\bfpage{933}--\blpage{949}
(\byear{2018})
\end{barticle}
\endbibitem

\bibitem[\protect\citeauthoryear{G{\"o}ing-Jaeschke and Yor}{2003}]{GoingJaeschkeYor2003}
\begin{barticle}
\bauthor{\bsnm{G{\"o}ing-Jaeschke}, \binits{A.}},
\bauthor{\bsnm{Yor}, \binits{M.}}:
\batitle{A survey and some generalizations of bessel processes}.
\bjtitle{Bernoulli}
\bvolume{9}(\bissue{2}),
\bfpage{313}--\blpage{349}
(\byear{2003})
\end{barticle}
\endbibitem

\bibitem[\protect\citeauthoryear{Heston}{1993}]{heston1993closed}
\begin{barticle}
\bauthor{\bsnm{Heston}, \binits{S.L.}}:
\batitle{A closed-form solution for options with stochastic volatility with applications to bond and currency options}.
\bjtitle{The Review of Financial Studies}
\bvolume{6}(\bissue{2}),
\bfpage{327}--\blpage{343}
(\byear{1993})
\end{barticle}
\endbibitem

\bibitem[\protect\citeauthoryear{Hansen and Lunde}{2006}]{HansenLunde2006}
\begin{barticle}
\bauthor{\bsnm{Hansen}, \binits{P.R.}},
\bauthor{\bsnm{Lunde}, \binits{A.}}:
\batitle{Realized variance and market microstructure noise}.
\bjtitle{Journal of Business \& Economic Statistics}
\bvolume{24}(\bissue{2}),
\bfpage{127}--\blpage{161}
(\byear{2006})
\doiurl{10.1198/073500106000000071}
\end{barticle}
\endbibitem

\bibitem[\protect\citeauthoryear{It\^o and McKean}{1974}]{ItoMcKean1974}
\begin{bbook}
\bauthor{\bsnm{It\^o}, \binits{K.}},
\bauthor{\bsnm{McKean}, \binits{H.P.}}:
\bbtitle{Diffusion Processes and Their Sample Paths}.
\bsertitle{Grundlehren der mathematischen Wissenschaften},
vol. \bseriesno{125}.
\bpublisher{Springer},
\blocation{Berlin}
(\byear{1974}).
\bcomment{Reprint of the 1965 edition}
\end{bbook}
\endbibitem

\bibitem[\protect\citeauthoryear{Jacod et~al.}{2009}]{JacodEtAl2009}
\begin{barticle}
\bauthor{\bsnm{Jacod}, \binits{J.}},
\bauthor{\bsnm{Li}, \binits{Y.}},
\bauthor{\bsnm{Mykland}, \binits{P.A.}},
\bauthor{\bsnm{Podolskij}, \binits{M.}},
\bauthor{\bsnm{Vetter}, \binits{M.}}:
\batitle{Microstructure noise in the continuous case: The pre-averaging approach}.
\bjtitle{Stochastic Processes and their Applications}
\bvolume{119}(\bissue{7}),
\bfpage{2249}--\blpage{2276}
(\byear{2009})
\doiurl{10.1016/j.spa.2008.11.004}
\end{barticle}
\endbibitem

\bibitem[\protect\citeauthoryear{Kim and Omberg}{1996}]{KimOmberg1996}
\begin{barticle}
\bauthor{\bsnm{Kim}, \binits{T.S.}},
\bauthor{\bsnm{Omberg}, \binits{E.}}:
\batitle{Dynamic nonmyopic portfolio behavior}.
\bjtitle{The Review of Financial Studies}
\bvolume{9}(\bissue{1}),
\bfpage{141}--\blpage{161}
(\byear{1996})
\end{barticle}
\endbibitem

\bibitem[\protect\citeauthoryear{Kraft}{2005}]{Kraft2005}
\begin{barticle}
\bauthor{\bsnm{Kraft}, \binits{H.}}:
\batitle{Optimal portfolios and heston's stochastic volatility model: An explicit solution for power utility}.
\bjtitle{Quantitative Finance}
\bvolume{5}(\bissue{3}),
\bfpage{303}--\blpage{313}
(\byear{2005})
\doiurl{10.1080/14697680500149503}
\end{barticle}
\endbibitem

\bibitem[\protect\citeauthoryear{Karatzas and Shreve}{1991}]{KaratzasShreve1991}
\begin{bbook}
\bauthor{\bsnm{Karatzas}, \binits{I.}},
\bauthor{\bsnm{Shreve}, \binits{S.E.}}:
\bbtitle{Brownian Motion and Stochastic Calculus},
\bedition{2nd} edn.
\bsertitle{Graduate Texts in Mathematics},
vol. \bseriesno{113}.
\bpublisher{Springer},
\blocation{New York}
(\byear{1991})
\end{bbook}
\endbibitem

\bibitem[\protect\citeauthoryear{Kraft and Steffensen}{2008}]{kraft2008optimal}
\begin{barticle}
\bauthor{\bsnm{Kraft}, \binits{H.}},
\bauthor{\bsnm{Steffensen}, \binits{M.}}:
\batitle{Optimal consumption and insurance: A continuous-time markov chain approach}.
\bjtitle{ASTIN Bulletin: The Journal of the IAA}
\bvolume{38}(\bissue{1}),
\bfpage{231}--\blpage{257}
(\byear{2008})
\doiurl{10.1017/S0515036100015154}
\end{barticle}
\endbibitem

\bibitem[\protect\citeauthoryear{Liu}{2007}]{Liu2007}
\begin{barticle}
\bauthor{\bsnm{Liu}, \binits{J.}}:
\batitle{Portfolio selection in stochastic environments}.
\bjtitle{The Review of Financial Studies}
\bvolume{20}(\bissue{1}),
\bfpage{1}--\blpage{39}
(\byear{2007})
\doiurl{10.1093/rfs/hhl001}
\end{barticle}
\endbibitem

\bibitem[\protect\citeauthoryear{Merton}{1969}]{Merton1969}
\begin{barticle}
\bauthor{\bsnm{Merton}, \binits{R.C.}}:
\batitle{Lifetime portfolio selection under uncertainty: The continuous-time case}.
\bjtitle{The Review of Economics and Statistics}
\bvolume{51}(\bissue{3}),
\bfpage{247}--\blpage{257}
(\byear{1969})
\doiurl{10.2307/1926560}
\end{barticle}
\endbibitem

\bibitem[\protect\citeauthoryear{Merton}{1971}]{Merton1971}
\begin{barticle}
\bauthor{\bsnm{Merton}, \binits{R.C.}}:
\batitle{Optimum consumption and portfolio rules in a continuous-time model}.
\bjtitle{Journal of Economic Theory}
\bvolume{3}(\bissue{4}),
\bfpage{373}--\blpage{413}
(\byear{1971})
\doiurl{10.1016/0022-0531(71)90038-X}
\end{barticle}
\endbibitem

\bibitem[\protect\citeauthoryear{Merton}{1973}]{Merton1973}
\begin{barticle}
\bauthor{\bsnm{Merton}, \binits{R.C.}}:
\batitle{An intertemporal capital asset pricing model}.
\bjtitle{Econometrica}
\bvolume{41}(\bissue{5}),
\bfpage{867}--\blpage{887}
(\byear{1973})
\doiurl{10.2307/1913811}
\end{barticle}
\endbibitem

\bibitem[\protect\citeauthoryear{Richard}{1975}]{Richard1975}
\begin{barticle}
\bauthor{\bsnm{Richard}, \binits{S.F.}}:
\batitle{Optimal consumption, portfolio and life insurance rules for an uncertain lived individual in a continuous time model}.
\bjtitle{Journal of Financial Economics}
\bvolume{2}(\bissue{2}),
\bfpage{187}--\blpage{203}
(\byear{1975})
\doiurl{10.1016/0304-405X(75)90004-5}
\end{barticle}
\endbibitem

\bibitem[\protect\citeauthoryear{Shreve and Soner}{1994}]{ShreveSoner1994}
\begin{barticle}
\bauthor{\bsnm{Shreve}, \binits{S.E.}},
\bauthor{\bsnm{Soner}, \binits{H.M.}}:
\batitle{Optimal investment and consumption with transaction costs}.
\bjtitle{The Annals of Applied Probability}
\bvolume{4}(\bissue{3}),
\bfpage{609}--\blpage{692}
(\byear{1994})
\doiurl{10.1214/aoap/1177004966}
\end{barticle}
\endbibitem

\bibitem[\protect\citeauthoryear{Vallejo-Jim{\'e}nez and Venegas-Mart{\'\i}nez}{2017}]{VallejoJimenez2017}
\begin{barticle}
\bauthor{\bsnm{Vallejo-Jim{\'e}nez}, \binits{B.}},
\bauthor{\bsnm{Venegas-Mart{\'\i}nez}, \binits{F.}}:
\batitle{Closed-form consumption--investment rules under markov-modulated preferences}.
\bjtitle{Economics Bulletin}
\bvolume{37}(\bissue{1}),
\bfpage{230}--\blpage{239}
(\byear{2017})
\end{barticle}
\endbibitem

\bibitem[\protect\citeauthoryear{Vallejo-Jim{\'e}nez et~al.}{2022}]{VenegasMartinez2022}
\begin{barticle}
\bauthor{\bsnm{Vallejo-Jim{\'e}nez}, \binits{B.}},
\bauthor{\bsnm{Venegas-Mart{\'\i}nez}, \binits{F.}},
\bauthor{\bsnm{Torre}, \binits{O.V.}},
\bauthor{\bsnm{{\'A}lvarez~Garc{\'\i}a}, \binits{J.}}:
\batitle{Simulating portfolio decisions under uncertainty when the risky asset and short rate are modulated by an inhomogeneous and asset-dependent markov chain}.
\bjtitle{Mathematics}
\bvolume{10}(\bissue{16}),
\bfpage{2926}
(\byear{2022})
\doiurl{10.3390/math10162926}
\end{barticle}
\endbibitem

\bibitem[\protect\citeauthoryear{Vallejo-Jim{\'e}nez et~al.}{2015}]{VallejoVenegasSoriano2015}
\begin{barticle}
\bauthor{\bsnm{Vallejo-Jim{\'e}nez}, \binits{B.}},
\bauthor{\bsnm{Venegas-Mart{\'\i}nez}, \binits{F.}},
\bauthor{\bsnm{Soriano-Morales}, \binits{Y.V.}}:
\batitle{Optimal consumption and portfolio decisions when the risky asset is driven by a time-inhomogeneous markov modulated diffusion process}.
\bjtitle{International Journal of Pure and Applied Mathematics}
\bvolume{104}(\bissue{3}),
\bfpage{353}--\blpage{362}
(\byear{2015})
\end{barticle}
\endbibitem

\bibitem[\protect\citeauthoryear{Yuan and Ao}{2012}]{YuanAo2012}
\begin{botherref}
\oauthor{\bsnm{Yuan}, \binits{R.}},
\oauthor{\bsnm{Ao}, \binits{P.}}:
Beyond it{\^o} versus stratonovich.
Journal of Statistical Mechanics: Theory and Experiment,
07010
(2012)
\doiurl{10.1088/1742-5468/2012/07/P07010}
\end{botherref}
\endbibitem

\end{thebibliography}
\endgroup

\end{document}